\documentclass[a4paper]{article}

\usepackage[T1]{fontenc}
\usepackage{lmodern}
\usepackage{amsmath,amssymb,amsthm}
\usepackage[top=3cm, bottom=3cm, left=3.5cm, right=3.5cm]{geometry}
\usepackage{microtype}
\usepackage{cleveref}

\theoremstyle{plain}
\newtheorem{mytheorem}               {Theorem}
\newtheorem{mylemma}[mytheorem]      {Lemma}
\newtheorem{myproposition}[mytheorem]{Proposition}
\newtheorem{mycorollary}[mytheorem]  {Corollary}
\theoremstyle{definition}
\newtheorem{mydefinition}[mytheorem] {Definition}
\newtheorem{myexample}[mytheorem]    {Example}
\theoremstyle{remark}
\newtheorem{myclaim}                 {Claim}

\crefname{mytheorem}    {Theorem}    {Theorems}
\crefname{mylemma}      {Lemma}      {Lemmas}
\crefname{mycorollary}  {Corollary}  {Corollaries}
\crefname{myproposition}{Proposition}{Propositions}
\crefname{mydefinition} {Definition} {Definitions}
\crefname{myexample}    {Example}    {Examples}

\crefname{section}      {Section}    {Sections}
\crefname{page}         {Page}       {Pages}

\title{Computational Complexity of the Minimum Cost Homomorphism Problem on Three-Element Domains}

\author{Hannes Uppman\thanks{Partially supported by the National Graduate School in Computer Science (CUGS), Sweden.}\\
{\small Department of Computer and Information Science,} \\
{\small Link{\"{o}}ping University, SE-581 83 Link{\"{o}}ping, Sweden} \\
{\small {\tt hannes.uppman@liu.se}}}
\date{}

\makeatletter
\newcommand{\st} {\@ifnextchar{~}{s.t.}{s.t.\ }}
\newcommand{\ie} {\@ifnextchar{~}{i.e.}{i.e.\ }}
\newcommand{\eg} {\@ifnextchar{~}{e.g.}{e.g.\ }}
\newcommand{\wlg}{\@ifnextchar{~}{wlog}{wlog\ }}
\newcommand{\wrt}{\@ifnextchar{~}{wrt.}{wrt.\ }}
\renewcommand{\iff}{\@ifnextchar{~}{iff}{iff\ }}
\makeatother

\renewcommand{\phi}{\varphi}
\renewcommand{\rho}{\varrho}
\renewcommand{\epsilon}{\varepsilon}

\newcommand{\Qplus}{\mathbb{Q}_{\ge 0}}
\newcommand{\Q}{\mathbb{Q}}

\newcommand{\Qplusinf}{\Qplus \cup \{\infty\}}

\DeclareMathOperator*{\argmin}{\arg\,\min}

\DeclareMathOperator{\ar}{ar}
\DeclareMathOperator{\pr}{pr}
\DeclareMathOperator{\supp}{supp}
\DeclareMathOperator{\sol}{Sol}

\DeclareMathOperator{\optsol}{Optsol}
\DeclareMathOperator{\feas}{Feas}
\DeclareMathOperator{\pol}{Pol}
\DeclareMathOperator{\fpol}{f{}Pol}

\newcommand{\csp}{CSP}
\newcommand{\vcsp}{VCSP}
\newcommand{\minhom}{Min-Cost-Hom}
\newcommand{\minsol}{Min-Sol}
\newcommand{\minones}{Min-Ones}

\newcommand{\class}[1]{#1}

\newcommand{\specialsol}[3]{\phi^{#1}_{#2 \rightarrow #3}}

\def\clap#1{\hbox to 0pt{\hss#1\hss}}
\def\mathllap{\mathpalette\mathllapinternal}
\def\mathrlap{\mathpalette\mathrlapinternal}

\def\mathllapinternal#1#2{\llap{$\mathsurround=0pt#1{#2}$}}
\def\mathrlapinternal#1#2{\rlap{$\mathsurround=0pt#1{#2}$}}

\newcommand{\close}[1] {\langle #1 \rangle}
\newcommand{\wclose}[1]{\langle #1 \rangle_w}
\newcommand{\eclose}[1]{\langle #1 \rangle_e}

\newcommand{\opers}[1]{\mathcal O_D^{(#1)}}
\newcommand{\comm}[3]{\left.\begin{smallmatrix}#1\\#2\end{smallmatrix}\right|_{\begin{smallmatrix}#3\end{smallmatrix}}}

\newcommand{\Zivny}{{\v{Z}}ivn{\'{y}}}
\newcommand{\TZ}{Thapper and \Zivny}
\newcommand{\Jonsson}{Jonsson}
\newcommand{\PJonsson}{Peter \Jonsson}
\newcommand{\Kolmogorov}{Kolmogorov}

\newcommand{\HKP}{Huber, Krokhin and Powell}
\newcommand{\Uppman}{Uppman}
\newcommand{\Takhanov}{Takhanov}
\newcommand{\Kuivinen}{Kuivinen}
\newcommand{\JKN}{\Jonsson, \Kuivinen\ and Nordh}
\newcommand{\FV}{Feder and Vardi}

\begin{document}
\maketitle
\begin{abstract}
In this paper we study the computational complexity of the (extended) \emph{minimum cost homomorphism problem} (\minhom) as a function of a constraint language, i.e.~a set of constraint relations and cost functions that are allowed to appear in instances.
A wide range of natural combinatorial optimisation problems can be expressed as {\minhom}s and a classification of their complexity would be highly desirable, both from a direct, applied point of view as well as from a theoretical perspective.

\minhom\ can be understood either as a flexible optimisation version of the \emph{constraint satisfaction problem} (\csp) or a restriction of the (general-valued) \emph{valued constraint satisfaction problem} (\vcsp).
Other optimisation versions of {\csp}s such as the \emph{minimum solution problem} (\minsol) and the \emph{minimum ones problem} (\minones) are special cases of \minhom.

The study of {\vcsp}s has recently seen remarkable progress.
A complete classification for the complexity of finite-valued languages on arbitrary finite domains has been obtained \TZ~[STOC'13].
However, understanding the complexity of languages that are not finite-valued appears to be more difficult.
\minhom\ allows us to study problematic languages of this type without having to deal with with the full generality of the \vcsp.
A recent classification for the complexity of three-element \minsol, \Uppman~[ICALP'13], takes a step in this direction.
In this paper we extend this result considerably by determining the complexity of three-element \minhom.
\end{abstract}

\section {Introduction}
The \emph{constraint satisfaction problem} (\csp) is a decision problem where 
an instance consists of a set of variables, a set of values, and a collection of constraints expressed over the variables.
The objective is to determine if it is possible to assign values to the variables in such a way that all constrains are satisfied simultaneously.
In general the constraint satisfaction problem is \class{NP}-complete.
However, by only allowing constraint-relations from a fixed constraint language $\Gamma$ one can obtain tractable fragments.
A famous conjecture by \FV~\cite{federvardi} predicts that this restricted problem, denoted \csp$(\Gamma)$, is either (depending on $\Gamma$) in \class{P} or is \class{NP}-complete.

In this paper we will study an optimisation version of the \csp.
Several such variants have been investigated in the literature.
Examples are: the \emph{min ones problem} (\minones)~\cite{booleancsp},
the \emph{minimum solution problem} (\minsol)~\cite{intromaxsol}
and the \emph{valued constraint satisfaction problem} (\vcsp)~\cite{vcsp}.
The problem we will work with is called the \emph{(extended) minimum cost homomorphism problem} (\minhom).
The ``unextended'' version of this problem was, motivated by a problem in defence logistics, introduced in~\cite{minhom:1} and studied in a series of papers before its complexity was completely characterised in~\cite{rustem:minhom}.
The extended version of the problem was introduced in~\cite{rustem:minsol}.

\minhom\ is a more general framework than both \minones\ and \minsol;
a problem of one of the latter types is also a \minhom.
The \vcsp-framework on the other hand is more general than \minhom.
In fact, we can describe every \minhom\ as a \vcsp\ for a constraint language in which every cost function is either $\{0,\infty\}$-valued or unary.
\minhom\ captures, despite this restriction, a wealth of combinatorial optimisation problems arising in a broad range of fields.

The study of {\vcsp}s has recently seen remarkable progress;
\TZ~\cite{tz:lp} described when a certain linear programming relaxation solves instances of the problem,
\Kolmogorov~\cite{kolmogorov:binary} simplified this description for finite-valued languages,
\HKP~\cite{annaco} classified all finite-valued languages on three-element domains, and
\TZ~\cite{tz:fvcsp} found a complete classification of the complexity for finite-valued languages on arbitrary finite domains.

Most of the classifications that have been obtained are about finite-valued constraint languages (\cite{tz:lp} mentioned above being a notable exception).
Understanding the complexity of general languages appears to be more difficult.
\minhom\ allows us to study languages of this type without having to deal with with the full generality of the \vcsp.
Using techniques of the so called algebraic approach (see \eg~\cite{alg:3,alg:1} and \cite{alg:opt:1,alg:opt:2}),
and building on results by \Takhanov~\cite{rustem:minhom,rustem:minsol} and \TZ~\cite{tz:lp,tz:fvcsp}
we could in \cite{minsol3} take a step in this direction by proving a classification for the complexity of \minsol\ on the three-element domain.
In this paper we extend these results to \minhom. 
Namely, we prove the following theorem.
\begin{mytheorem}
\label{r:main}
Let $(\Gamma,\Delta)$ be a finite language on a three-element domain $D$
and define $\Gamma^+ = \Gamma \cup \{\{d\}:d\in D\} \cup \{\{x : \nu(x)<\infty\} : \nu \in \Delta\}$.
If $(\Gamma,\Delta)$ is a core, then one of the following is true.
\begin{itemize}
\item \minhom$(\Gamma^+,\Delta)$ can be proved to be in \class{PO} by \cref{r:semilattice}.
\item \minhom$(\Gamma^+,\Delta)$ can be proved to be in \class{PO} by \cref{r:gwtp}.
\item \minhom$(\Gamma,\Delta)$ is \class{NP}-hard.
\end{itemize}
\end{mytheorem}

We define cores in \cref{s:core}.
\cref{r:main} combined with the following result, which follows from \cite[Lemma~2.4]{tz:fvcsp}, yields a full classification for \minhom\ on three-elements.
\begin{myproposition}
If $(\Gamma',\Delta')$ is a core of $(\Gamma,\Delta)$ then
\minhom$(\Gamma,\Delta)$ and \minhom$(\Gamma',\Delta')$ are polynomial-time inter-reducible.
\end{myproposition}

To obtain the classification we apply tools from the algebraic approach, and, following \TZ,
we make repeated use of Motzkin's Theorem.
Our tractability results are formulated and proved for arbitrary finite domains and are therefore not restricted to the three-element case.
Many of the tools we derive to aid in proving our main theorem are also effective on domains of size larger than three.
One example is that we show that a relation fails to be in the wpp-closure of a language only if some fractional polymorphism of the language does not preserve the relation (\cref{r:subset:fpol}).
This complements results in~\cite{alg:opt:1,alg:opt:2}.
Another example is that we show that all constants can be added to a core language without significantly changing the complexity of the associated \minhom (\cref{r:get:const}).
This complements results in~\cite{tz:fvcsp}.

The rest of the paper is organised as follows.
In \cref{s:prelim} we define some fundamental concepts,
in \cref{s:tractable} we state and prove tractability results,
in \cref{s:tools} we collect a number of results that will be used later on (these might also be useful on domain of larger size),
in \cref{s:core} we define cores~\cite{tz:fvcsp} and prove a related result,
in \cref{s:three:el} we focus on the three-element domains and establish our main result; that core languages that are not tractable by the results in \cref{s:tractable} are in fact \class{NP}-hard,
and finally, in \cref{s:proofs}, we give proofs for results stated in \cref{s:tools} and \cref{s:three:el}.

\section {Preliminaries}
\label{s:prelim}
Let $D$ be a finite set.
The pair $(\Gamma,\Delta)$ 
is called a finite \emph{language} if $\Gamma$ is a finite set of finitary relations on $D$ and 
$\Delta$ is a finite set of functions $D \to \Qplusinf$.
For every finite language $(\Gamma,\Delta)$ we define the optimisation problem \minhom$(\Gamma,\Delta)$ as follows.
\begin{description}
\item[Instance:] A triple $(V,C,w)$ where
\begin{itemize}
\item
$V$ is a set of variables,
\item
$C$ is a set of $\Gamma$-allowed constraints, \ie a set of pairs $(s,R)$ where the constraint-scope $s$ is a tuple of variables, and the constraint-relation $R$ is a member of $\Gamma$ of the same arity as $s$,
\item
$w$ is a weight function $V \times \Delta \to \Qplus$.
\end{itemize}
\item[Solution:] A function $\phi : V \to D$ \st for every $(s,R) \in C$ it holds that $\phi(s) \in R$, where $\phi$ is applied component-wise.
\item[Measure:] The measure of a solution $\phi$ is $m(\phi) = \sum_{v \in V} \sum_{\nu \in \Delta} w(v,\nu) \nu( \phi(v) )$.
For every function $\phi : V \to D$ that is not a solution we define $m(\phi)=\infty$.
\end{description}
The objective is to find a solution $\phi$ that minimises $m(\phi)$.

For an instance $I$ we let $\sol(I)$ denote the set of all solutions and $\optsol(I)$ the set of all optimal solutions.
We define $0\,\infty = \infty\,0 = 0$, $x \le \infty$ and $x + \infty = \infty + x = \infty$ for all $x \in \Qplusinf$.

\subsection {Notation}
The $i$:th projection operation will be denoted $\pr_i$.
We define $\binom A 2 = \{ \{x,y\} \subseteq A : x \ne y \}$.
The set of operations on $D$ is denoted $\mathcal{O}_D$.
For binary operations $f$, $g$ and $h$ we define $\overline f$ through $\overline f(x,y) = f(y,x)$ and $f[g,h]$ through $f[g,h](x,y) = f(g(x,y),h(x,y))$.
A $k$-ary operation $f$ on $D$ is called \emph{conservative} if $f(x_1,\dots,x_k) \in \{x_1,\dots,x_k\}$ for every $x_1,\dots,x_k \in D$.
A ternary operation $m$ on $D$ is called \emph{arithmetical} (or \emph{2/3-minority})
if $m(x,y,y)=m(x,y,x)=m(y,y,x)=x$ for every $x,y \in D$.
We say that an operation $f$ on $D$ is conservative (arithmetical) on $S \subseteq D$ if $f|_S$ is conservative (arithmetical).
Similarly we say that $f$ is conservative (arithmetical) on $\mathcal{S} \subseteq 2^D$ if $f|_S$ is conservative (arithmetical) for every $S \in \mathcal{S}$.

For a set $A$ of operations (relations) we write $A^{(k)}$ for the set of all $k$-ary operations (relations) in $A$.
For a set $\Gamma$ of relations on $D$ we use $\Gamma^c$ to denote $\Gamma \cup \{\{d\}: d \in D\}$.

We use $\delta$ for the Kronecker delta function, \ie $\delta_{x,y}=1$ if $x=y$ and $\delta_{x,y}=0$ otherwise.

\subsection {Polymorphisms}
An function $f : D^m \to D$ is called a \emph{polymorphism} of $\Gamma$ if for every $R \in \Gamma$ and every $t_1,\dots,t_m \in R$ it holds that $f(t_1,\dots,t_m)\in R$  where $f$ is applied component-wise.
The set of all polymorphisms of $\Gamma$ is denoted $\pol(\Gamma)$.
A function $\omega : \pol^{(k)}(\Gamma) \to \Qplus$ is a $k$-ary \emph{fractional polymorphism}~\cite{alg:opt:1} of $(\Gamma,\Delta)$ \iff $\sum_{g\in \pol^{(k)}(\Gamma)} \omega(g)=1$ and
\begin{align*}
&
\sum_{g\in \pol^{(k)}(\Gamma)} \omega(g) \nu(g(x_1,\dots, x_k)) \le \frac{1}{k} \sum_{i=1}^k \nu(x_i)
&\nu \in \Delta, x_1,\dots,x_k \in D.
\end{align*}
The support of a fractional polymorphism $\omega$, denoted $\supp(\omega)$, if the set of polymorphisms for which $\omega$ is non-zero.
The set of all fractional polymorphisms of $(\Gamma,\Delta)$ is denoted $\fpol(\Gamma,\Delta)$.
\begin{myexample}
The function $\pr_i$ is a trivial polymorphism for any set of relations $\Gamma$, and
the function $f \mapsto \sum_{i=1}^k \frac{1}{k} \delta_{\pr_i,f}$ is a $k$-ary fractional polymorphism of every language $(\Gamma,\Delta)$.
\end{myexample}

\subsection {Reductions}
A relation 
$R$ is called \emph{pp-definable} in $\Gamma$ \iff there is an instance $I=(V,C)$ of \csp$(\Gamma)$ \st
$R = \{ (\phi(v_1),\dots,\phi(v_n)) : \phi \in \sol(I) \}$ for some $v_1,\dots,v_n \in V$.
The notation $\close{\Gamma}$ is used for the set of all relations that are pp-definable in $\Gamma$.
Similarly; 
$R$ is called \emph{weighted pp-definable} (wpp-definable) in $(\Gamma,\Delta)$ \iff there is an instance $I=(V,C,w)$ of \minhom$(\Gamma,\Delta)$ \st
$R = \{ (\phi(v_1),\dots,\phi(v_n)) : \phi \in \optsol(I) \}$ for some $v_1,\dots,v_n \in V$.
We use $\wclose{\Gamma,\Delta}$ to denote the set of all such relations.
A function $\nu : D \to \Qplusinf$ is called \emph{expressible} in $(\Gamma,\Delta)$ \iff there is an instance
$I=(V,C,w)$ of \minhom$(\Gamma,\Delta)$ and $v \in V$ \st
$\nu(x) = \min \{ m(\phi) : \phi : V \to D, \phi(v)=x \}$.
The set of all cost functions expressible in $(\Gamma,\Delta)$ is denoted $\eclose{\Gamma,\Delta}$.
We use $\feas(\Delta)$ for the set $\{ \{ x : \nu(x)<\infty \} : \nu \in \Delta \}$.

What makes these closures interesting is the following result, see \eg~\cite{alg:opt:1,softcsp,fourel}.
\begin{mytheorem}
\label{r:red}
Let $\Gamma' \subseteq \wclose{\Gamma,\Delta}$ and $\Delta' \subseteq \eclose{\Gamma,\Delta}$ be finite sets.
Then, it holds that
\minhom$(\Gamma' \cup \feas(\Delta'),\Delta')$
is polynomial-time reducible to \minhom$(\Gamma,\Delta)$.
\end{mytheorem}

\section {Tractable languages}
\label{s:tractable}
We will make use of two tractability results.
The first
follows from
a theorem by \TZ~\cite[Theorem~5.1 (see remarks in Sect.~6)]{tz:lp}.
\begin{mytheorem}
\label{r:semilattice}
Let $(\Gamma,\Delta)$ be a finite language.
If there exists $\omega \in \fpol^{(2)}(\Gamma,\Delta)$ with $f \in \supp(\omega)$ \st $f$ is a semilattice operation,
then \minhom$(\Gamma, \Delta)$ is in \class{PO}.
\end{mytheorem}
\begin{myexample}
\label{ex:smf}
Let $(\Gamma,\Delta)$ be a language on a totally ordered domain $D$ that admits the binary fractional polymorphism
$f \mapsto \frac{1}{2}\delta_{\min,f} + \frac{1}{2}\delta_{\max,f}$.
Certainly $\min$ is a semilattice operation, so by \cref{r:semilattice} it follows that \minhom$(\Gamma,\Delta)$ is in \class{PO}.
\end{myexample}
We remark that the theorem in \cite{tz:lp} from which \cref{r:semilattice} follows is very capable;
it explains the tractability of every finite-valued \vcsp\ that is not \class{NP}-hard~\cite{tz:fvcsp}.

The second tractability result generalises a family of languages that \Takhanov\ has proved tractable~\cite{rustem:minhom,rustem:minsol}.
The particular formulation we will use here is a bit more general than a version we previously used in~\cite[Theorem~8]{minsol3}.

To state the result we need to introduce a few concepts.
A central observation is given by the following lemma.
The result follows immediately from the definition of fractional polymorphisms and the measure function $m$.
We omit the proof.
\begin{mylemma}
\label{r:observation}
If $(\Gamma,\Delta)$ admits a $k$-ary fractional polymorphism $\omega$
and $I$ is an instance of \minhom$(\Gamma,\Delta)$ with $\phi_1,\dots,\phi_k \in \sol(I)$,
then $f(\phi_1,\dots,\phi_k) \in \sol(I)$ for every $f \in \supp(\omega)$ and
\begin{align*}
\sum_{f \in \pol^{(k)}(\Gamma)} \omega(f) m(f(\phi_1,\dots,\phi_k)) \le \frac{1}{k} \sum_{i=1}^k m(\phi_k).
\end{align*}
\end{mylemma}
\begin{myexample}
\label{ex:smf:2}
Consider again \cref{ex:smf}.
It follows from \cref{r:observation} that, for any instance $I=(V,C,w)$
and any $\phi_1,\phi_2 : V \to D$,
we have $m(\min(\phi_1,\phi_2)) + m(\max(\phi_1,\phi_2)) \le m(\phi_1) + m(\phi_2)$.
Functions of this kind are called \emph{submodular} and are central characters in the field of discrete optimisation, see \eg~\cite{sfm}.
\end{myexample}
The following two definitions establishes some convenient notation.
\begin{mydefinition}
\label{def:w}
For functions $\omega \in \fpol^{(k)}(\Gamma,\Delta)$ and $x \in D, y \in D^k$ we define
$W^\omega_x(y) = \sum_{f \in \pol^{(k)}(\Gamma) : f(y)=x} \omega(f)$.
When there is no risk of confusion we drop the superscript and write $W_x(y)$.
\end{mydefinition}
\begin{mydefinition}
\label{def:special:sol}
For an instance $I=(V,C,w)$ of \minhom$(\Gamma,\Delta)$, a variable $v \in V$
and a value $x \in \{ \phi(v) : \phi \in \sol(I)\}$,
we denote by $\specialsol{I}{v}{x}$ an arbitrary solution of $I$ \st
$m(\specialsol{I}{v}{x}) = \min \{ m(\phi) : \phi \in \sol(I), \phi(v) = x \}$.
\end{mydefinition}
Using these definitions we obtain the following corollary of \cref{r:observation}.
\begin{mylemma}
\label{r:measure}
If $(\Gamma,\Delta)$ admits a $k$-ary fractional polymorphism $\omega$,
$I=(V,C,w)$ is an instance of \minhom$(\Gamma,\Delta)$ and $v \in V$ is \st $\{a_1,\dots,a_k\} \subseteq \{ \phi(v) : \phi \in \sol(I)\}$,
then
\begin{align*}
\sum_{d \in D} W_d(a_1,\dots,a_k) m(\specialsol{I}{v}{d}) \le \frac{1}{k} \sum_{i=1}^k m(\specialsol{I}{v}{a_i}).
\end{align*}
\end{mylemma}
\begin{mydefinition}
We say that $S \subseteq D$ is \emph{shrinkable} to $S \setminus \{x\}$ in $(\Gamma,\Delta)$
if $(\Gamma,\Delta)$ admits a sequence of fractional polymorphisms $\omega_1,\dots,\omega_m$ 
and tuples $a^1 \in S^{\ar(\omega_1)},\dots,a^m \in S^{\ar(\omega_m)}$
\st 
for an instance $I=(V,C,w)$ of \minhom$(\Gamma,\Delta)$ and $v \in V$ \st $S \subseteq \{\phi(v) : \phi \in \sol(I)\}$
it holds that
the system of inequalities we obtain from \cref{r:measure} applied to $\omega_i$ and $a_i$, for $i\in[m]$,
implies that
\begin{align*}
\sum_{i=1}^n t_i m(\specialsol{I}{v}{a_i}) \le m(\specialsol{I}{v}{x})
\end{align*}
for some integer $n$, some $t_1,\dots,t_n \in \Qplus$ \st $\sum_{i=1}^n t_i = 1$, and some $a_1,\dots,a_n \in S\setminus \{x\}$.

We call a collection of fractional polymorphisms and tuples of this type a \emph{certificate} for the fact that $D$ is shrinkable to $D \setminus \{x\}$.
If $S$ is shrinkable to $S \setminus\{x\}$ and $S \setminus\{x\}$ is shrinkable to $S \setminus\{x,y\}$,
then we say that $S$ is shrinkable to $S \setminus\{x,y\}$.
\end{mydefinition}
\begin{myexample}
\label{ex:csp:shrinkable}
Consider the language $(\Gamma,\emptyset)$ on the domain $D$.
Let $\{a_1,\dots,a_m\} \subseteq D$.
It is not hard to see that $\omega : f \mapsto \sum_{i=1}^{m-1} \frac{1}{m-1} \delta_{\pr_i,f}$ is in $\fpol^{(m)}(\Gamma,\emptyset)$.
Hence, $\omega$ and $(a_1,\dots,a_{m})$ certifies that $\{a_1,\dots,a_m\}$ is shrinkable to $\{a_1,\dots,a_{m-1}\}$.
\end{myexample}

We can now state the second tractability result.
\begin{mytheorem}
\label{r:gwtp}
Let $(\Gamma,\Delta)$ be a finite language on the domain $D$ \st $\Gamma = \Gamma^c$ and \st \csp$(\Gamma)$ is in $\class{P}$.
\minhom$(\Gamma, \Delta)$ is in \class{PO} if there exits
$\mathcal{F} \subseteq \wclose{\Gamma,\Delta}^{(1)}$, $\mathcal{A} \subseteq \binom D 2$, $f_1,f_2 \in \pol^{(2)}(\Gamma)$ and $m \in \pol^{(3)}(\Gamma)$ \st the following holds.
\begin{itemize}
\item
If $\{a,b\} \subseteq B$ for some $B \in \mathcal{F}$, and $\{a,b\} \not\in \mathcal{A}$, then $f_1|_{\{a,b\}}$ and $f_2|_{\{a,b\}}$ are projections
and $m|_{\{a,b\}}$ is arithmetical.
\item
If $\{a,b\} \subseteq B$ for some $B \in \mathcal{F}$, and $\{a,b\} \in \mathcal{A}$, then $f_1|_{\{a,b\}}$ and $f_2|_{\{a,b\}}$ are different idempotent, conservative and commutative operations.
\item
For every $S \in \wclose{\Gamma,\Delta}^{(1)} \setminus \mathcal{F}$ there is a certificate showing that $S$ is shrinkable to some $S' \in \mathcal{F}$.
\item
$m$ is idempotent on every set in $\mathcal{F}$
and conservative on every set in $\mathcal{F}\setminus\mathcal{A}$.
\end{itemize}
\end{mytheorem}
\begin{proof}[Proof sketch]
Given an instance $I$ of \minhom$(\Gamma,\Delta)$ we can, since \csp$(\Gamma^c)$ is in $\class{P}$,
compute for every variable $v$ the set $D_v = \{ \phi(v) : \phi \in \sol(I) \}$.
From the definition of shrinkable sets it is immediate that if $D_v$ is shrinkable to $S \in \wclose{\Gamma,\Delta}$,
then we can add the constraint $(v,S)$ to $I$ without worsening the measure of an optimal solution.
We can repeat this procedure until $D_v$ is in $\mathcal{F}$ for every variable $v$.

It is known, see~\cite[Proof of Theorem~8]{minsol3}, that
from $f_1,f_2,m$ one can construct (by superposition) operations $f_1,f_2,m$ that 
in addition to the
conditions of the theorem
also satisfy the following stronger properties:
\begin{itemize}
\item
If $\{a,b\} \subseteq B$ for some $B \in \mathcal{F}$
and $\{a,b\}\not\in \mathcal{A}$, then $f_1'|_{\{a,b\}} = f_2'|_{\{a,b\}} = \pr_1$.
\item
The operation $m'$ is idempotent and conservative on every set in $\mathcal{F}$.
\end{itemize}
Clearly $f_1',f_2',m' \in \pol(\Gamma)$.
Note that $f_1',f_2',m'$ preserves all unary relations $S \subseteq B$ for $B \in \mathcal{F}$.
The result therefore follows from an easy reduction to the multi-sorted version of the problem and a result due to 
\Takhanov\ for this conservative multi-sorted variant~\cite[Theorem~23]{rustem:minsol}.
\end{proof}

\begin{myexample}
Consider again \minhom$(\Gamma,\emptyset)$.
We saw in \cref{ex:csp:shrinkable} that for every $\{x\} \subseteq X \subseteq D$ it holds that $X$ is shrinkable to $\{x\}$.
Hence, if $\Gamma^c = \Gamma$ and \csp$(\Gamma)$ is in \class{P} it follows from \cref{r:gwtp} that \minhom$(\Gamma,\emptyset)$ is in \class{PO}.
This of course is no surprise as \minhom$(\Gamma,\emptyset)$ essentially is the same problem as \csp$(\Gamma)$.
\end{myexample}

\section {Tools}
\label{s:tools}

In this section we establish a few results that will come in handy later on.
Most of these results are used in proofs collected in \cref{s:proofs}.
However, we hope this section will provide an overview of the kind of techniques that are used to prove our main theorem.
Several of the results are proved with the help of the following classical theorem, see \eg~\cite[p.~94]{schrijver}.
\begin{mytheorem}[Motzkin's Transposition Theorem]
\label{r:motzkin}
For any
$A \in \Q^{m \times n}$,
$B \in \Q^{p \times n}$,
$b \in \Q^{m}$ and
$c \in \Q^{p}$, exactly one of the following holds:
\begin{itemize}
\item $A x \le b$, $Bx < c$ for some $x \in \Q^n$
\item $A^T y + B^T z = 0$ and ($b^T y + c^T z < 0$ or $b^T y + c^T z = 0$ and $z \ne 0$) for some $y \in \Qplus^m$ and $z \in \Qplus^p$
\end{itemize}
\end{mytheorem}

The first result concerns a slight generalisation of the concept of dominating fractional polymorphisms~\cite{minsol3}.
\begin{mydefinition}
\label{def:domp:fpol}
Let $k \ge 2$ and $a \in D^{k-1}$, $b \in D$ be \st $a_1,\dots,a_{k-1},b$ are distinct elements.
A fractional polymorphism $\omega \in \fpol^{(k)}(\Gamma,\Delta)$ is called \emph{$(a_1,\dots,a_{k-1},b)$-dominating} if
$W^\omega_{a_j}(a_1,\dots,a_{k-1},b) \ge \frac{1}{k}$
for every $j \in [k-1]$ and 
$\frac{1}{k} > W^\omega_{b}(a_1,\dots,a_{k-1},b)$.
\end{mydefinition}
\begin{myproposition}
\label{r:use:dom}
Let $(\Gamma,\Delta)$ be a finite language on a finite set $D$.
Let $k \ge 2$ and $a \in D^{k-1}$, $b \in D$ be \st $a_1,\dots,a_{k-1},b$ are distinct.
If $(\Gamma,\Delta)$ does not admit a fractional polymorphism that is $(a_1,\dots,a_{k-1},b)$-dominating, then $\eclose{\Gamma,\Delta}$ contains a unary function $\nu$ that satisfies $\infty > \nu(a_1),\dots,\nu(a_{k-1}),\nu(b)$
and $\nu(c)>\nu(b)$ for every $c \in D \setminus \{b\}$.
\end{myproposition}
A proof is given in \cref{s:r:use:dom}.
Using similar arguments we can also prove the following characterisation of which relations that are wpp-definable in $(\Gamma,\Delta)$.
\begin{myproposition}
\label{r:subset:fpol}
Let $(\Gamma,\Delta)$ be a finite language on a finite set $D$ and let $\emptyset \ne R = \{t_1,\dots,t_k\} \subseteq D^n$.
Exactly one of the following is true.
\begin{enumerate}
\item
There exists $\omega \in \fpol^{(k)}(\Gamma,\Delta)$ with $f \in \supp(\omega)$ \st $f(t_1,\dots,t_k) \not\in \{t_1,\dots,t_k\}$.
\item
It holds that $R \in \wclose{\Gamma,\Delta}$.
\end{enumerate}
\end{myproposition}

We give a proof in \cref{s:r:subset:fpol}.
Once established we can use the proposition to quickly derive a number of useful results.

\begin{mycorollary}
\label{r:wpp:computable}
Let $(\Gamma,\Delta)$ be a finite language on a finite set $D$.
For any fixed $k$ the set of wpp-definable $k$-ary relations, $\wclose{\Gamma,\Delta}^{(k)}$, can be computed.
\end{mycorollary}
\begin{proof}[Proof sketch]
This is immediate from \cref{r:subset:fpol};
we can find all polymorphisms of arities $1,\dots,|D|^k$ and then, for every $R \subseteq D^k$, solve a linear program.
\end{proof}

\begin{mycorollary}
\label{r:two:subset}
Let $(\Gamma,\Delta)$ be a finite language on a finite set $D$ and let $\{a,b\} \subseteq D$.
If there is $\nu \in \eclose{\Gamma,\Delta}$ and $A \subseteq D$ \st $\{a,b\} \subseteq A$, $A \in \wclose{\Gamma,\Delta}$ and $\nu(a)<\nu(b)<\infty$ and $\nu(b) \le \nu(x)$ for any $x \in A\setminus\{a,b\}$, then one of the following is true.
\begin{enumerate}
\item
$\{a,b\} \in \wclose{\Gamma,\Delta}$
\item
There is $\omega \in \fpol^{(2)}(\Gamma,\Delta)$ that is $(a,b)$-dominating.
\end{enumerate}
\end{mycorollary}
\begin{proof}
Assume (1) does not hold.
By \cref{r:subset:fpol} there must exist some $\omega \in \fpol^{(2)}(\Gamma,\Delta)$ with $f \in \supp(\omega)$ \st $f(a,b) \not\in \{a,b\}$.
It is not hard to see that in this case, because of $\nu$, the fractional polymorphism $\omega$ must be $(a,b)$-dominating.
Hence, (2) must be true.
\end{proof}

\begin{mycorollary}
\label{r:no:dom}
Let $(\Gamma,\Delta)$ be a finite language on a finite set $D$ and let $\{a_1,\dots,a_k\} \subseteq D$.
One of the following is true.
\begin{enumerate}
\item
There is 
$\omega \in \fpol^{(k)}(\Gamma,\Delta)$ and $i \in [k]$ \st $\omega$ is $(a_1,\dots,a_{i-1},$ $a_{i+1},\dots,a_k,a_i)$-dominating.
\item
For every $i \in [k]$ there is $j \in [k]\setminus\{i\}$ \st $\{a_i,a_j\} \in \wclose{\Gamma,\Delta}$.
\end{enumerate}
\end{mycorollary}
\begin{proof}
Assume (1) is false.
By \cref{r:use:dom}, for any $i \in [k]$, there is $\nu_i \in \eclose{\Gamma,\Delta}$ \st
$\argmin_{x \in D} \nu_i(x) = \{a_i\}$ and $\nu_i(x)<\infty$ if $x \in \{a_1,\dots,a_k\}$.
Let $i \in [m]$.
Pick $j$ \st $\nu_i(a_j) = \min\{ \nu_i(x) : x \in \{a_1,\dots,a_{i-1},a_{i+1},\dots,a_k\} \}$.

Note that there is no $\psi \in \fpol^{(2)}(\Gamma,\Delta)$ that is $(a_i,a_j)$-dominating;
if there was then
\begin{align*}
f \mapsto \sum_{i=1}^{k-2} \frac{1}{k} \delta_{\pr_i,f} +
\sum_{g \in \supp(\psi)} \frac{2}{k} \psi(g) \delta_{g[\pr_{k-1},\pr_k], f}
\end{align*}
would be $(x_1,\dots,x_{k-2},a_i,a_j)$-dominating for $x_1,\dots,x_{k-2} \in D$.
Hence, by \cref{r:two:subset},
we have $\{a_i,a_j\} \in \wclose{\Gamma,\Delta}$.
Since the choice of $i$ was arbitrary (2) must be true.
\end{proof}

The generalised min-closed languages were introduced by \JKN~\cite{maxonesgen}
and defined as sets of relations preserved by a particular type of binary operation.
\Kuivinen~\cite[Section 5.5]{kuivinen} provides an alternative characterisation of the languages as those preserved by a so called min set function.

A \emph{set function}~\cite{width:one} is a function $f:2^D\setminus \{\emptyset\} \to D$.
A \emph{$\nu$-min set function}~\cite{kuivinen} is a set function $f$
satisfying $\nu(f(X)) \le \min\{\nu(x) : x \in X\}$ for every
$X \in 2^D\setminus \{\emptyset\}$.
The following proposition, which is a variant of \cite[Theorem~5.18]{kuivinen}, will later prove to be useful.
\begin{myproposition}
\label{r:min:set}
Let $(\Gamma,\{\nu\})$ be a finite language \st $\wclose{\Gamma,\{\nu\}}^{(1)} \subseteq \Gamma$.
The following are equivalent:
\begin{enumerate}
\item
$\Gamma$ is preserved by a $\nu$-min set function,
\item
$\Gamma$ is preserved by a set function $f$ \st $\nu(f(X)) = \min\{\nu(x): x\in\bigcap_{Y\in\close{\Gamma}: Y\supseteq X} Y\}$
for every $X \in 2^D \setminus \{\emptyset\}$,
\item
$\Gamma$ is preserved by a set function and for every $R \in \close{\Gamma}$ it holds that
\begin{align*}
R \cap ( \argmin_{x \in \pr_1(R)} \nu(x) \times \dots \times \argmin_{x \in \pr_{\ar(R)}(R)} \nu(x) ) &\ne \emptyset.
\end{align*}
\newcounter{tempenum}
\setcounter{tempenum}{\theenumi}
\end{enumerate}
Furthermore, if $\nu$ is injective, then the following condition is equivalent to the ones above.
\begin{enumerate}
\setcounter{enumi}{\thetempenum}
\item
For every $R \in \close{\Gamma}$ it holds that
\begin{align*}
R \cap ( \argmin_{x \in \pr_1(R)} \nu(x) \times \dots \times \argmin_{x \in \pr_{\ar(R)}(R)} \nu(x) ) &\ne \emptyset.
\end{align*}
\end{enumerate}
\end{myproposition}
The proof, which we for the sake of completeness state in \cref{s:r:min:set}, is similar to that in~\cite{kuivinen}.

Let $\nu : D \to \Qplus$ be injective.
We call the binary relation $R$ a \emph{cross} (with respect to $\nu$) \iff $|R|\ge 2$ and there are $\alpha_1,\alpha_2 \in \Q_{>0}$ \st 
$\alpha_1 \nu(t_1) + \alpha_2 \nu(t_2) = 1$ for every $t \in R$.
The following lemma is a generalisation of \cite[Lemma~25]{minsol3}.

\begin{mylemma}
\label{r:cross}
Let $\nu : D \to \Qplus$ be injective.
If $\Gamma$ is not preserved by a $\nu$-min set function, then $\wclose{\Gamma,\Delta}$ contains a cross.
\end{mylemma}
\begin{proof}
If $\Gamma$ is not preserved by a $\nu$-min set function,
then \cref{r:min:set} implies that there is $R \in \close{\Gamma}$ \st
$(\min_\nu(\pr_1(R)), \dots, \min_\nu(\pr_{\ar(R)}(R))) \not\in R$.

In fact, there must be a binary relation in $\close{\Gamma}$ of this kind.
To see this let $R \in \close{\Gamma}$ be a $k$-ary relation \st $(\min_\nu(\pr_1(R)), \dots, \min_\nu(\pr_k(R))) \not\in R$
and \st that every relation $R' \in \close{\Gamma}$ of smaller arity satisfies
$(\min_\nu(\pr_1(R')), \dots, \min_\nu(\pr_{\ar(R')}(R'))) \in R'$.
This means that there is $t^1 \in R$ \st $t^1_i = \min_\nu(\pr_i(R))$ for $i \in [k]\setminus\{1\}$,
otherwise $\pr_{2,\dots,\ar(R)}(R)$ contradicts the minimality of $k$.
Similarly there is $t^2 \in R$ \st $t^2_i = \min_\nu(\pr_i(R))$ for $i \in [k]\setminus\{2\}$.
This means that $R' = \{(x,y) : (x,y,\min_\nu(\pr_3(R)),\dots,\min_\nu(\pr_k(R))) \in R \}$ is a non-empty relation of arity 2 \st
$(\min_\nu(\pr_1(R')), \min_\nu(\pr_2(R'))) \not\in R'$.
Hence, $k=2$.

Clearly we can choose $\alpha_1,\alpha_2$ \st $R'' = \argmin_{(x,y) \in R} ( \alpha_1\nu(x)+\alpha_2\nu(y) )$ satisfies $|R''| \ge 2$,
and $R'' \in \wclose{\Gamma,\Delta}$ is a cross.
\end{proof}

To prove that a given language is computationally hard we make use of the following lemma
which is an immediate consequence of~\cite[Theorem~3.1]{rustem:minhom}.
\begin{mylemma}
\label{r:np:hard}
If $\{a,b\} \in \Gamma$ and $\nu(a) < \nu(b) < \infty$, $\sigma(b) < \sigma(a) < \infty$ for some $\nu,\sigma \in \Delta$,
then either
\begin{itemize}
\item
there exists $f,g \in \pol^{(2)}(\Gamma)$ \st $f|_{\{a,b\}}$ and $g|_{\{a,b\}}$ are two different idempotent, commutative and conservative operations,
\item
there exists $m \in \pol^{(3)}(\Gamma)$ \st $m|_{\{a,b\}}$ is arithmetical, or
\item
\minhom$(\Gamma,\Delta)$ are both \class{NP}-hard.
\end{itemize}
\end{mylemma}
The following result by \Takhanov~\cite[Theorem 5.4]{rustem:minhom} shows how ``partially arithmetical'' polymorphisms
(like the ones that we might get out of the previous lemma) can be stitched together.
\begin{mylemma}
\label{r:arithmetical}
Let $C \subseteq \binom D 2$.
If $C \subseteq \Gamma$ and for each $\{a,b\} \in C$ an operation in $\pol^{(3)}(\Gamma)$ is arithmetical on $\{a,b\}$,
then there is an operation in $\pol^{(3)}(\Gamma)$ that is arithmetical on $C$.
\end{mylemma}

The next lemma is a
variation, see~\cite[Lemma~14]{minsol3}, of a lemma by \TZ~\cite[Lemma~3.5]{tz:fvcsp}.
It
allows us 
to prove the existence of certain nontrivial fractional polymorphisms.
We may also obtain this lemma as a simple corollary of \cref{r:subset:fpol}.
\begin{mylemma}
\label{r:fpol}
If $\{(a,b),(b,a)\} \not\in \wclose{\Gamma,\Delta}$, then for every $\sigma \in \eclose{\Gamma,\Delta}$ there is $\omega \in \fpol(\Gamma,\Delta)$ with $f \in \supp(\omega)$ \st $\{f(a,b),f(b,a)\} \ne \{a,b\}$ and $\sigma(f(a,b))+\sigma(f(b,a)) \le \sigma(a)+\sigma(b)$.
\end{mylemma}

We will make use the following notation.
\begin{mydefinition}
Let $P \subseteq \opers 2$.
For a function $\omega : P \to \Qplus$ we define $\omega^2 : P \to \Qplus$ by
$\omega^2(f) = \sum_{g,h \in P: g[h,\overline{h}]=f} \omega(g) \omega(h)$.
\end{mydefinition}
Regarding the above construction we note the following.
A proof is given in \cref{s:r:mul:is:fpol}.
\begin{mylemma}
\label{r:mul:is:fpol}
If $\omega \in \fpol^{(2)}(\Gamma,\Delta)$, then $\omega^2 \in \fpol(\Gamma,\Delta)$.
\end{mylemma}

Finally, the following two lemmas, which are proved in \cref{s:r:min:exists} and \cref{s:r:min:commutative},
are used to ``canonicalise'' interesting fractional polymorphisms.
\begin{mylemma}
\label{r:min:exists}
Let $\beta : D^2 \to \Qplus$
and define 
$C_\omega(x) = \sum_{f \in \pol^{(2)}(\Gamma) : f(x) = \overline{f}(x)} \omega(f)$ and $M(\omega) = \sum_{x \in D^2} C_{\omega}(x)$.
Set $\Omega = \{ \omega \in \fpol^{(2)}(\Gamma,\Delta) : \forall s \in D^2, C_\omega(s) \ge \beta(s) \}$.
If 
$\wclose{\Gamma,\Delta}^{(1)} \subseteq \Gamma$,
then either $\Omega = \emptyset$, or there is $\omega^* \in \Omega$ \st $M(\omega^*) = \sup_{\omega \in \Omega} M(\omega)$.
\end{mylemma}
\begin{mylemma}
\label{r:min:commutative}
Let $S \subseteq \binom D 2$ and
$\Pi = \{ \omega \in \fpol^{(2)}(\Gamma,\Delta) : $ for all $s \in S$ there exists $f \in \supp(\omega)$ \st $f|_{s}$ is commutative$\}$.
If
$\wclose{\Gamma,\Delta}^{(1)} \subseteq \Gamma$ and $\Pi \ne \emptyset$,
then there is $\omega \in \Pi$ \st for every $f \in \supp(\omega)$ and $x \in D^2$ it holds that $\{f(x),\overline{f}(x)\} \not\in S$.
\end{mylemma}

\section {Cores}
\label{s:core}
In this section we define cores and prove that one can add all constants to a language that is a core without making the associated \minhom\ much more difficult.
We use a definition of cores from \cite[Definition~3]{tz:fvcsp}.

\begin{mydefinition}
A finite language $(\Gamma,\Delta)$ is a \emph{core} \iff for every $\omega \in \fpol^{(1)}(\Gamma,\Delta)$
and every $f \in \supp(\omega)$ it holds that $f$ is injective.
A language $(\Gamma',\Delta')$ is a core of another language $(\Gamma,\Delta)$ if $(\Gamma',\Delta')$ is a core and
$(\Gamma',\Delta') = (\Gamma,\Delta)|_{g(D)}$ for some $\psi \in \fpol^{(1)}(\Gamma,\Delta)$ and $g \in \supp(\psi)$.
\end{mydefinition}

A result very similar to the following was given in \cite{annaco,tz:fvcsp} for finite-valued languages.
\begin{myproposition}
\label{r:get:const}
If $(\Gamma,\Delta)$ is a core,
then \minhom$(\Gamma^c,\Delta)$ is polynomial-time reducible to \minhom$(\Gamma,\Delta)$.
\end{myproposition}
\begin{proof}[Proof sketch]
We will show that \minhom$(\Gamma^c,\Delta)$ is polynomial-time reducible to \minhom$(\Gamma \cup \wclose{\Gamma,\Delta}^{(|D|)}, \Delta)$.
By \cref{r:red} this is sufficient.

Assume $D = \{d_1,\dots,d_{|D|}\}$.
Let $R = \{ (d_1,\dots,d_{|D|}) \}$ and let $R'$ be the closure of $R$ under the operations $f \in \supp(\omega)$, $\omega \in \fpol^{(1)}(\Gamma,\Delta)$.

Note that there is no $k>1$, $\psi \in \fpol^{(k)}(\Gamma,\Delta)$ and $g \in \supp(\psi)$ \st $g$ does not preserve $R'$.
This follows from the fact that $R'$ was generated from a single tuple.
It is not hard to show that there is $\varpi \in \fpol^{(1)}(\Gamma,\Delta)$ \st
$R' = \{ f(d_1,\dots,d_{|D|}) : f \in \supp(\varpi) \}$.
Assume that there is $s=f(t^1,\dots,t^k) \not\in R'$ for some $f \in \supp(\psi)$ and $t^1,\dots,t^k \in R'$.
This means that we from $\psi$ and $\varpi$ can construct
$\varpi' \in \fpol^{(1)}(\Gamma,\Delta)$ with
$f \in \supp(\varpi')$ \st $s=f(d_1,\dots,d_{|D|})$, which is a contradiction.

From \cref{r:subset:fpol} it follows that $R' \in \wclose{\Gamma,\Delta}$.
Since $(\Gamma,\Delta)$ is a core, for every $\omega \in \fpol^{(1)}(\Gamma,\Delta)$ and $f \in \supp(\omega)$
we know that $f$ is injective.
Hence, every $t \in R'$ equals $(\pi(d_1),\dots,\pi(d_{|D|}))$ for some permutation $\pi$ on $D$.

We now use a construction that is applied for the corresponding result for {\csp}s~\cite[Theorem 4.7]{alg:3}.
Given an instance $I$ of \minhom$(\Gamma^c,\Delta)$ we create an instance of $I'$ of
\minhom$(\Gamma \cup \wclose{\Gamma,\Delta}^{(|D|)},\Delta)$ from $I$ by adding variables $v_{d_1},\dots,v_{d_{|D|}}$
and replacing every constraint $(v,\{d_i\})$ with the constraint $((v,v_{d_i}),=)$.
Finally we add the constraint $((v_{d_1},\dots,v_{d_{|D|}}), R')$.
If there is a solution to $I$, then there is also a solution to $I'$.
And, if $\psi$ is an optimal solution to $I'$, then
$\phi(v_{d_1},\dots,v_{d_{|D|}}) = (\pi(d_1),\dots,\pi(d_{|D|}))$
for some permutation $\pi$ on $D$ and $\omega \in \fpol^{(1)}(\Gamma,\Delta)$ \st $\pi \in \supp(\omega)$.
Hence $\pi^k \circ \psi$ is another optimal solution to $I'$, for any $k \ge 1$.
In particular there is an optimal solution $\phi^*$ to $I'$ \st $\phi^*(v_{d_1},\dots,v_{d_{|D|}}) = (d_1,\dots,d_{|D|})$.
This allows us to recover an optimal solution to $I$.
\end{proof}

\section {Proof of \cref{r:main}}
\label{s:three:el}
In this section we establish a sequence of lemmas that together imply our main result.
To save ink we begin by giving short names to a few statements.
\begin{description}
\item[A1:]
$(\Gamma,\Delta)$ is a finite language on $D=\{a,b,c\}$ \st $\Gamma^c \cup \feas(\Delta) \cup \wclose{\Gamma,\Delta}^{(1)} \cup \wclose{\Gamma,\Delta}^{(2)} \subseteq \Gamma$.
\item[G1:]
\minhom$(\Gamma,\Delta)$ can be shown to be in \class{PO} by \cref{r:gwtp}.
\item[G2:]
\minhom$(\Gamma,\Delta)$ can be shown to be in \class{PO} by \cref{r:semilattice}.
\item[G3:]
\minhom$(\Gamma,\Delta)$ is \class{NP}-hard.
\end{description}

The supporting lemma below is used to show the results that follow.
We give a proof in \cref{s:r:get:valuations}.
\begin{mylemma}
\label{r:get:valuations}
Assume A1.
If $\{a,b\} \not\in \Gamma$,
then either there is $\omega \in \fpol^{(2)}(\Gamma,\Delta)$ that is $(a,b)$ or $(b,a)$-dominating,
or there are $\nu_a,\nu_b \in \eclose{\Gamma,\Delta}$ \st $\nu_a(a) < \nu_a(c) < \nu_a(b)$
and $\nu_b(b) < \nu_b(c) < \nu_b(a)$.
\end{mylemma}

We are going to analyse a few different cases depending on the number of two-element subsets of the domain that is wpp-definable in $(\Gamma,\Delta)$.
The following lemma, which follows immediately from \cref{r:no:dom}, connects this number to dominating fractional polymorphisms.
\begin{mylemma}
\label{r:get:two:subsets}
Assume A1.
Either $|\Gamma \cap \binom D 2| \ge 2$ or there is 
$\omega \in \fpol^{(3)}(\Gamma,\Delta)$ and $a_1,a_2,a_3 \in D$ \st $\omega$ is $(a_1,a_2,a_3)$-dominating and $\{a_1,a_2,a_3\}=D$.
\end{mylemma}
To understand languages that admit a ternary dominating fractional polymorphism we use the following lemma. We give a proof in \cref{s:r:two:subsets}.
\begin{mylemma}
\label{r:two:subsets}
Assume A1.
If $\{a,b\} \not\in \Gamma$
and there is $\omega \in \fpol^{(3)}(\Gamma,\Delta)$ \st $\omega$ is $(a,b,c)$-dominating, then either $\{a,c\},\{b,c\} \in \Gamma$,
or G1, G2 or G3 is true.
\end{mylemma}
The following four lemmas are used to handle languages that contain two unary two-element relations.
We prove them in \cref{s:r:twoset:dom,s:r:work:a,s:r:class:one,s:r:work:b}.
\begin{mylemma}
\label{r:twoset:dom}
Assume A1.
If
$\{a,c\},\{c,b\} \in \Gamma$
and
there is $\omega \in \fpol^{(2)}(\Gamma,\Delta)$ that is $(a,b)$-dominating, then G1 or G3 is true.
\end{mylemma}
\begin{mylemma}
\label{r:work:a}
Assume A1.
If $\{a,b\} \not\in \Gamma$
and $\{a,c\},\{c,b\} \in \Gamma$,
then either $\{(a,c),(c,a)\} \in \Gamma$,
$\{(b,c),(c,b)\} \in \Gamma$,
or G1, G2, or G3 is true.
\end{mylemma}
\begin{mylemma}
\label{r:class:one}
Assume A1.
If $\{a,b\} \not\in \Gamma$,
$\{a,c\},\{c,b\} \in \Gamma$ and
$\{(a,c),(c,a)\} \in \Gamma$
and $\{(b,c),(c,b)\} \not\in \Gamma$,
then G1 or G3 holds.
\end{mylemma}
\begin{mylemma}
\label{r:work:b}
Assume A1.
If $\{a,b\} \not\in \Gamma$
and $\{(a,c),(c,a)\},\{(b,c),(c,b)\} \in \Gamma$,
then G1 or G3 holds.
\end{mylemma}

We can now prove the main theorem.
\begin{proof}[Proof of \cref{r:main}]
Let $\Gamma'=\wclose{\Gamma,\Delta}^{(1)} \cup \wclose{\Gamma,\Delta}^{(2)} \cup \Gamma^c \cup \feas(\Delta)$.
Note that if \minhom$(\Gamma',\Delta)$ can be shown to be in \class{PO}
using \cref{r:semilattice} or \cref{r:gwtp}, then so can \minhom$(\Gamma^c \cup \feas(\Delta),\Delta)$.
Furthermore, by \cref{r:red} and \cref{r:get:const} we know that 
\minhom$(\Gamma',\Delta)$ is polynomial time reducible to \minhom$(\Gamma,\Delta)$.
Hence, if \minhom$(\Gamma',\Delta)$ is \class{NP}-hard, then so is \minhom$(\Gamma,\Delta)$.

Clearly, if \csp$(\Gamma')$ is \class{NP}-hard, then so is \minhom$(\Gamma',\Delta)$.
And, if \csp$(\Gamma')$ is not \class{NP}-hard, then it is in \class{PO}.
This follows from~\cite{bulatov}.

If $|\binom D 2 \cap \Gamma'|=3$, then \cref{r:gwtp} can place \minhom$(\Gamma',\Delta)$
in \class{PO} unless \minhom$(\Gamma',\Delta)$ is \class{NP}-hard.
This follows from~\cite[Theorem~12]{minsol3}.

If $|\binom D 2 \cap \Gamma'|<2$, then, by \cref{r:get:two:subsets}, we know that there is $\omega \in \fpol^{(3)}(\Gamma',\Delta)$
that is $(a_1,a_2,a_3)$-dominating for some $\{a_1,a_2,a_3\}=D$.
If $\{a_1,a_2\} \not\in \Gamma'$, then by \cref{r:two:subsets} we know that either $|\binom D 2 \cap \Gamma'|=2$ (a contradiction) or \minhom$(\Gamma',\Delta)$ can be proved to be in \class{PO} by either \cref{r:gwtp} or \cref{r:semilattice}, or \minhom$(\Gamma',\Delta)$ is \class{NP}-hard.
Otherwise $\{a_1,a_2\} \in \Gamma'$.
Since $|\binom D 2 \cap \Gamma'|<2$ it must hold that $\{a_1,a_3\} \not\in \Gamma'$ and $\{a_2,a_3\} \not\in \Gamma'$.
In this case, since $\{a_1,a_2,a_3\}$ is shrinkable to $\{a_1,a_2\}$, it holds that either
\minhom$(\Gamma',\Delta)$ can be proved to be in \class{PO} by either \cref{r:gwtp} or \minhom$(\Gamma',\Delta)$ is \class{NP}-hard.

The only remaining case is $|\binom D 2 \cap \Gamma'|=2$.
In this case the result follows from \cref{r:work:a}, \cref{r:class:one} and \cref{r:work:b}.
\end{proof}

\section* {Acknowledgements}
I am thankful to \PJonsson\ for rewarding discussions.

\newpage
\appendix

\section {Proofs}
\label{s:proofs}

We will use the following notation.
\begin{mydefinition}
For $x,y,z \in D$ we define $\comm x y z = \{ f \in \opers 2 : f(x,y)=f(y,x)=z \}$.
Similarly, for $x,y,z \in D^m$ we define
$\comm{x_1 &\cdots &x_m}{y_1 &\cdots &y_m}{z_1 & \cdots &z_m} = \comm{x_1}{y_1}{z_1} \cap \cdots \cap\comm{x_m}{y_m}{z_m}$.
\end{mydefinition}

\subsection{Proof of \cref{r:use:dom}}
\label{s:r:use:dom}
\newcommand\PP{P^{(k)}}
\newcommand\DD{D^k_\nu}
Let $k \ge 2$ and $a \in D^{k-1}$, $b \in D$ be \st $a_1,\dots,a_{k-1},b$ are distinct.

For $\nu : D \to \Qplusinf$ define $\nu^k : D^k \to \Qplusinf$ by $\nu^k(x_1,\dots,x_k) = \frac{1}{k} \sum_{i=1}^k \nu(x_i)$.
Set
$\DD = \{x \in D^k : \nu^k(x)<\infty\}$,
$\PP = \{g \in \pol^{(k)}(\Gamma) : \nu(g(x))<\infty \text{ for every } \nu \in \Delta \text{ and } x \in \DD \}$, 
and, for $x \in D$, let $\PP_x = \{g \in \PP : g(a_1,\dots,a_{k-1},b)=x\}$.

It is not hard to see that 
the language $(\Gamma,\Delta)$ admits a $(a_1,\dots,a_{k-1},b)$-dominating $k$-ary fractional polymorphism \iff the following system has a solution $(u_g\in \Q : g\in \PP)$.
\begin{align*}
 \sum_{g \in \PP} u_g \nu( g(x) ) &\le \nu^k(x)  & \nu \in \Delta, x \in \DD \\
-u_g &\le 0                                      & g \in \PP\\ 
 \sum_{g \in \PP} u_g &\le  1 \\ 
-\sum_{g \in \PP} u_g &\le -1 \\
-\sum_{g \in \PP_{a_i}} u_g &\le -\frac{1}{k}        & i \in [k-1]\\
 \sum_{g \in \PP_b} u_g &<    \frac{1}{k}
\end{align*}
If the system is unsatisfiable, then, by \cref{r:motzkin}, there are $(v_{\nu,x} \in \Qplus : \nu \in \Delta, x \in \DD)$, $(o_g \in \Qplus : g \in \PP)$, $w_1, w_2 \in \Qplus$, $(y_{a_i} \in \Qplus : i \in [k-1])$ and $y_b \in \Qplus$ \st 
\begin{align*}
\sum_{\nu \in \Delta, x \in \DD} \nu(g(x)) v_{\nu,x} - o_g + w_1 - w_2 - y_{a_i} &= 0 &i \in [k-1], g \in \PP_{a_i}\\
\sum_{\nu \in \Delta, x \in \DD} \nu(g(x)) v_{\nu,x} - o_g + w_1 - w_2 + y_b &= 0 &g \in \PP_b\\
\sum_{\nu \in \Delta, x \in \DD} \nu(g(x)) v_{\nu,x} - o_g + w_1 - w_2       &= 0 &g \in \PP \setminus (\bigcup_{i=1}^{k-1} \PP_{a_i} \cup \PP_b)
\end{align*}
and 
\begin{align*}
\sum_{\nu \in \Delta, x \in \DD} \nu^k(x) v_{\nu,x}
+ w_1 - w_2 - \frac{1}{k} \sum_{i=1}^{k-1} y_{a_i} + \frac{1}{k} y_b &= \alpha,
\end{align*}
where either $\alpha < 0$ or $\alpha = 0$ and $y_b>0$.
Hence, for every $f_1, \dots, f_k \in \PP$ \st $(f_1,\dots,f_k)(a_1,\dots,a_{k-1},b)=(a_1,\dots,a_{k-1},b)$, we have
\begin{align*}
   \sum_{\nu \in \Delta, x \in \DD} \nu^k(x) v_{\nu,x} + \frac{1}{k} \sum_{i=1}^k o_{f_i}
&= \sum_{\nu \in \Delta, x \in \DD} \nu^k((f_1,\dots,f_k)(x)) v_{\nu,x} + \alpha.
\end{align*}
Note that since $\pr_1,\dots,\pr_k \in \PP$ and $(\pr_1,\dots,\pr_k)(a_1,\dots,a_{k-1},b)=(a_1,\dots,a_{k-1},b)$
we must have $\alpha=0$, $o_{\pr_i} = 0$ for $i \in [k]$, and $y_b>0$.
This means that the following is true.
\begin{align*}
&\min_{f \in \PP_{a_i}} \sum_{\nu \in \Delta, x \in \DD} \nu(f(x)) v_{\nu,x} = \sum_{\nu \in \Delta, x \in \DD} \nu(\pr_i(x)) v_{\nu,x} = - w_1 + w_2 + y_{a_i} &i \in [k-1]\\
&\min_{f \in \PP_b} \sum_{\nu \in \Delta, x \in \DD} \nu(f(x)) v_{\nu,x} = \sum_{\nu \in \Delta, x \in \DD} \nu(\pr_k(x)) v_{\nu,x} = - w_1 + w_2 - y_b \\
&\min_{f \in \PP \setminus (\bigcup_{i=1}^{k-1} \PP_{a_i} \cup \PP_b)} \sum_{\nu \in \Delta, x \in \DD} \nu(f(x)) v_{\nu,x} \ge - w_1 + w_2
\end{align*}

Create an instance $I$ of \minhom$(\Gamma,\Delta)$ with variables $D^k$ and measure
\begin{align*}
m(\phi) &= \sum_{\nu \in \Delta, x \in \DD} v_{\nu,x} \nu(\phi(x)) + \epsilon \sum_{\nu \in \Delta, x \in \DD} \nu(\phi(x)),
\end{align*}
where
$\epsilon \in \Q_{>0}$ is chosen small enough \st $\phi \in \argmin_{\phi' \in \PP} m(\phi')$ implies that $\phi \in \argmin_{\phi' \in \PP} \sum_{\nu \in \Delta, x \in \DD} v_{\nu,x} \nu(\phi(x))$.
Such a number $\epsilon$ can always be found.
Note that a solution $\phi$ to $I$ with finite measure is a function $D^k \to D$ \st $\nu(\phi(x))<\infty$ for every $\nu \in \Delta$ and $x \in \DD$.

Pick, for every $g \in \opers k \setminus \pol^{(k)}(\Gamma)$, a relation $R_g \in \Gamma$ \st $g$ does not preserve $R_g$.
Add for each $k$-sequence of tuples $t^1,\dots,t^k \in R_g$ the constraint $(( (t^1_1,\dots,t^k_1), \dots, (t^1_{\ar(R_g)},\dots,t^k_{\ar(R_g)})), R_g)$.
This construction is essentially the second order indicator problem~\cite{indicator}.
Now a solution to $I$ is by construction a $k$-ary polymorphism of $\Gamma$.
Hence, if $\phi$ is a solution to $I$ with finite measure, then $\phi \in \PP$.
Clearly $\pr_1,\dots,\pr_k$ satisfies all constraints and are solutions to $I$ with finite measure.
Since $y_{a_i} \ge 0$ and $y_b>0$, it also holds that $\min_{\phi \in \PP_x} m(\phi) > \min_{\phi \in \PP_{b}} m(\phi)$ for every $x \in D \setminus \{b\}$.
So, with $\nu(x) = \min_{g \in \sol(I) : g(a_1,\dots,a_{k-1},b)=x} m(g)$, we have $\infty > \nu(a_1),\dots,\nu(a_{k-1})$
and also $\nu(c) > \nu(b)$ for every $c \in D \setminus \{b\}$.
Since $\nu \in \eclose{\Gamma,\Delta}$ we are done.

\subsection {Proof of \cref{r:subset:fpol}}
\label{s:r:subset:fpol}
If (1) is true, then (2) must be false. In the rest of the proof we show that if (1) is false, then (2) is true.

For $\nu : D \to \Qplusinf$ define $\nu^k : D^k \to \Qplusinf$ by $\nu^k(x_1,\dots,x_k) = \frac{1}{k} \sum_{i=1}^k \nu(x_i)$.
Assume $R = \{t^1,\dots,t^k\}$.
Set
$\DD = \{x \in D^k : \nu^k(x)<\infty\}$ and
$\PP = \{g \in \pol^{(k)}(\Gamma) : \nu(g(x))<\infty \text{ for every } \nu \in \Delta \text{ and } x \in \DD \}$.
Define $\Omega = \{g \in \PP : g(t^1,\dots,t^k) \not\in R\}$.

It is not hard to see that 
there exists $\omega \in \fpol^{(k)}(\Gamma,\Delta)$ with $f \in \supp(\omega)$ \st $f \in \Omega$ \iff the following system has a solution $(u_g\in \Q : g\in \PP)$.
\begin{align*}
 \sum_{g \in \PP} u_g \nu( g(x) ) &\le \nu^k(x)  & \nu \in \Delta, x \in \DD \\
-u_g &\le 0                                      & g \in \PP\\ 
 \sum_{g \in \PP} u_g &\le  1 \\ 
-\sum_{g \in \PP} u_g &\le -1 \\
-\sum_{g \in \Omega} u_g &<0
\end{align*}
If the system is unsatisfiable, then, by \cref{r:motzkin}, there are $(v_{\nu,x} \in \Qplus : \nu \in \Delta, x \in \DD)$, $(o_g \in \Qplus : g \in \PP)$, $w_1, w_2 \in \Qplus$ and $y \in \Qplus$ \st 
\begin{align*}
\sum_{\nu \in \Delta, x \in \DD} \nu(g(x)) v_{\nu,x} - o_g + w_1 - w_2     &= 0 &g \in \PP \setminus \Omega\\
\sum_{\nu \in \Delta, x \in \DD} \nu(g(x)) v_{\nu,x} - o_g + w_1 - w_2 - y &= 0 &g \in \Omega
\end{align*}
and 
\begin{align*}
\sum_{\nu \in \Delta, x \in \DD} \nu^k(x) v_{\nu,x} + w_1 - w_2 &= \alpha,
\end{align*}
where either $\alpha < 0$ or $\alpha = 0$ and $y>0$.
Hence, for every $f_1,\dots,f_k \in \PP$ \st $(f_1,\dots,f_k)(t^1,\dots,t^k) = (t^1,\dots,t^k)$ (with functions applied component-wise), we have
\begin{align*}
   \sum_{\nu \in \Delta, x \in \DD} \nu^k(x) v_{\nu,x} + \frac{1}{k} \sum_{i=1}^k o_{f_i}
&= \sum_{\nu \in \Delta, x \in \DD} \nu^k((f_1,\dots,f_k)(x)) v_{\nu,x} + \alpha.
\end{align*}
Note that since $\pr_1,\dots,\pr_k \in \PP$ and $(\pr_1,\dots,\pr_k)(t^1,\dots,t^k)=(t^1,\dots,t^k)$
we must have $\alpha=0$, $o_{\pr_i} = 0$ for $i \in [k]$, and $y>0$.
This means that the following is true.
\begin{align*}
\min_{f \in \PP \setminus \Omega} \sum_{\nu \in \Delta, x \in \DD} \nu(f(x)) v_{\nu,x}
&= \sum_{\nu \in \Delta, x \in \DD} \nu(\pr_1(x)) v_{\nu,x} = - w_1 + w_2\\
\min_{f \in \Omega} \sum_{\nu \in \Delta, x \in \DD} \nu(f(x)) v_{\nu,x}
&\ge - w_1 + w_2 + y
\end{align*}

Create an instance $I$ of \minhom$(\Gamma,\Delta)$ with variables $D^k$ and measure
\begin{align*}
m(\phi) &= \sum_{\nu \in \Delta, x \in \DD} v_{\nu,x} \nu(\phi(x)) + \epsilon \sum_{\nu \in \Delta, x \in \DD} \nu(\phi(x)),
\end{align*}
where
$\epsilon \in \Q_{>0}$ is chosen small enough \st $\phi \in \argmin_{\phi' \in \PP} m(\phi')$ implies that $\phi \in \argmin_{\phi' \in \PP} \sum_{\nu \in \Delta, x \in \DD} v_{\nu,x} \nu(\phi(x))$.
Such a number $\epsilon$ can always be found.
Note that a solution $\phi$ to $I$ with finite measure is a function $D^k \to D$ \st $\nu(\phi(x))<\infty$ for every $\nu \in \Delta$ and $x \in \DD$.

Pick, for every $g \in \opers k \setminus \pol^{(k)}(\Gamma)$, a relation $R_g \in \Gamma$ \st $g$ does not preserve $R_g$.
Add for each $k$-sequence of tuples $t^1,\dots,t^k \in R_g$ the constraint $(( (t^1_1,\dots,t^k_1), \dots, (t^1_{\ar(R_g)},\dots,t^k_{\ar(R_g)})), R_g)$.
This construction is essentially the second order indicator problem~\cite{indicator}.
Now a solution to $I$ is by construction is a $k$-ary polymorphism of $\Gamma$.
Hence, if $\phi$ is a solution to $I$ with finite measure, then $\phi \in \PP$.
Clearly $\pr_1,\dots,\pr_k \in \PP \setminus \Omega$ satisfies all constraints and are solutions to $I$ with finite measure.
Since $y>0$ it holds that $\min_{\phi \in \Omega} m(\phi) > \min_{\phi \in \PP \setminus \Omega} m(\phi)$.
So
$\{ ( \phi(t^1_1,\dots,t^k_1), \dots, \phi(t^1_{\ar(R)},\dots,t^k_{\ar(R)}) ) : \phi \in \optsol(I) \} = R$ and we are done.

\subsection {Proof of \cref{r:min:set}}
\label{s:r:min:set}
It is easy to see that (2) implies (1).
Clearly (1) implies (3) as by definition there is a $\nu$-min set function
$f$ that preserves every $R \in \Gamma$, and therefore also every $R \in \close{\Gamma}$.

We now show that (3) implies (2).
For $S \in 2^D \setminus \{\emptyset\}$ let $U(S) = \bigcap \{ S' \in \close{\Gamma}^{(1)} : S' \supseteq S\}$.
Let $g$ be any set function that preserves $\Gamma$ (by (3) such a function must exist).
Define $f(S) = g( M(U(S)) )$ where $M(X) = \argmin_{x \in X} \nu(x)$.
Note that for all $S \in 2^D \setminus \{\emptyset\}$ it holds that $M(U(S)) \ne \emptyset$ since by (3) and by the fact that $U(S) \in \close{\Gamma}$ it holds that $U(S) \cap M( U(S) ) \ne \emptyset$.
It follows that $f$ is a set function.
Since $\wclose{\Gamma,\{\nu\}}^{(1)} \subseteq \Gamma$ and $g$ preserves $\Gamma$ it must hold that $f(S) \in M(U(S))$, so $\nu(f(S)) \in \{ \nu(x) : x \in M(U(S)) \}$.
What remains is to show that $f$ preserves $\Gamma$.

Let $R$ be a $n$-ary relation in $\Gamma$ and $P \subseteq R$.
Note that $R' = R \cap (U(\pr_1(P)) \times \dots \times U(\pr_n(P))) \in \close{\Gamma}$.
Note also that by construction $\pr_i(R') = U(\pr_i(P))$,
so by (3) we know
\begin{align*}
R'' 
&= R' \cap (M(\pr_1(R')) \times \dots \times M(\pr_n(R'))) \\
&= R' \cap (M(U(\pr_1(P))) \times \dots \times M(U(\pr_n(P))))
\ne \emptyset.
\end{align*}
Since $M(U(\pr_i(P))) \in \wclose{\Gamma,\{\nu\}}^{(1)} \subseteq \Gamma$ we have
$R'' \in \close{\Gamma}$, and $g$ must preserve $R''$.
Hence, 
\begin{align*}
(f(\pr_1(P)),\dots,f(\pr_n(P)) &= ( g( M(U(\pr_1(P)))),\dots, g( M(U(\pr_n(P))))) \\
                               &= ( g( \pr_1(R'')),\dots, g( \pr_n(R'')) )
                               \in R'' \subseteq R.
\end{align*}

Note that if $\nu$ is injective and (4) is true, then $M(U(S))$ is a one-element set.
This means that $f(S) = h(M(U(S)))$ where $h(\{x\})=x$ for every $x \in D$ is a set function that preserves every $R\in\close{\Gamma}$.
Hence, (3) is true.
Clearly (3) implies (4), so the proof is complete.

\newcommand{\Pt}{\pol^{(2)}(\Gamma)}

\subsection {Proof of \cref{r:mul:is:fpol}}
\label{s:r:mul:is:fpol}
Note that, since $g,h \in \Pt$ implies $g[h,\overline{h}] \in \Pt$,
\begin{align*}
\sum_{f \in \Pt} \omega^2(f)
&= 
\sum_{f \in \Pt}
\sum_{\substack{g, h \in \Pt:\\ g[h,\overline{h}] = f}} 
\omega(g) \omega(h)
= \sum_{g, h \in \Pt} \omega(g) \omega(h)
= 1
\end{align*}
and that
\begin{align*}
  \sum_{f \in \Pt} \omega^2(f) \nu( f(x,y) )
&=
  \sum_{f \in \Pt}
  \sum_{\substack{g, h \in \Pt:\\ g[h,\overline{h}] = f}} 
  \omega(g) \omega(h) \nu( g[h,\overline{h}](x,y) ) \\
&=
  \sum_{h \in \Pt} \omega(h) \sum_{g \in \Pt} \omega(g) \nu( g[h,\overline{h}](x,y) ) \\
&\le
  \sum_{h \in \Pt} \omega(h) \frac 1 2 ( \nu(h(x,y)) + \nu(\overline{h}(x,y)) ) \\
&=
  \frac{1}{2}
  \big( 
  \sum_{h \in \Pt} \omega(h) \nu( h(x,y) )
  +
  \sum_{h \in \Pt} \omega(h) \nu( h(y,x) )
  \big) \\
&\le
  \frac 1 2 ( \nu(x) + \nu(y) ).
\end{align*}

\subsection {Proof of \cref{r:min:exists}}
\label{s:r:min:exists}
Unless $\Omega = \emptyset$ we can pick $w^*$ as the function $f \mapsto u_f$ given by the optimal solution to the following linear program.
\begin{flalign*}
&\mathrlap{\text{maximise } \sum_{x \in D^2, g\in \Pt : g(x) = \overline{g}(x)} u_g}&&                   &\\
&\rlap{subject to}& u_g &\ge 0                                                                           &\mathllap{g \in \Pt}\\ 
&                 & \sum_{g\in \Pt} u_g &=  1                                                            &\\
&                 & \sum_{g\in \Pt : \nu(g(x,y))<\infty} u_g \nu(g(x,y)) &\le \frac{1}{2}(\nu(x)+\nu(y)) &\mathllap{\begin{aligned}x,y \in D, \nu \in \Delta:&\\ \nu(x),\nu(y) < \infty&\end{aligned}} \\
&                 & \sum_{g\in \Pt : \nu(g(x,y))=\infty} u_g &=0                                         &\mathllap{\begin{aligned}x,y \in D, \nu \in \Delta:&\\ \nu(x),\nu(y) < \infty&\end{aligned}} \\
&                 & \sum_{g\in \Pt : g(x) = \overline{g}(x)} u_g &\ge \beta(x)                           &\mathllap{x \in D^2}
\end{flalign*}
An optimal solution to this finite and bounded program clearly exists.

\subsection {Proof of \cref{r:min:commutative}}
\label{s:r:min:commutative}
Pick any $\omega \in \Pi$.
Define $\beta : D^2 \to \Qplus$ as follows.
Set $\beta(x,y) = C_{\omega}(x,y)$ if $\{x,y\}=s$ for some $s \in S$, otherwise $\beta(x,y)=0$.
If follows from \cref{r:min:exists} that there is some $\omega^* \in \Omega$ that maximises $M$ (with $\Omega$, $M$ as defined in \cref{r:min:exists}).

Assume there is $p \in \supp(\omega^*)$, $q \in D^2$ and $s \in S$ \st $\{p(q),\overline{p}(q)\}=s$.
Note that
\begin{align*}
M((\omega^*)^2)
&=
  \sum_{x \in D^2} C_{(\omega^*)^2}(x) \\
&=
  \sum_{x \in D^2}
  \sum_{\substack{f \in \Pt:\\ f(x)=\overline{f}(x)}}
  \sum_{\substack{g,h \in \Pt:\\ g[h,\overline{h}] = f}} \omega^*(g) \omega^*(h) \\
&=
  \sum_{x \in D^2}
  \sum_{\substack{g, h \in \Pt:\\ g((h,\overline{h})(x)) = \overline{g}((h,\overline{h})(x))}} \omega^*(g) \omega^*(h) \\
&=
  \sum_{x \in D^2} \sum_{\substack{h \in \Pt:\\h(x)=\overline{h}(x)}}
  \sum_{g \in \Pt} \omega^*(g) \omega^*(h) 
  +
  \sum_{x \in D^2} \sum_{\substack{h \in \Pt:\\h(x)\ne\overline{h}(x)}}
  \sum_{\substack{g \in \Pt:\\ g((h,\overline{h})(x)) = \overline{g}((h,\overline{h})(x))}} \omega^*(g) \omega^*(h) \\
&\ge
  \sum_{x \in D^2} \sum_{\substack{h \in \Pt:\\h(x)=\overline{h}(x)}}
  \sum_{g \in \Pt} \omega^*(g) \omega^*(h) 
  +
  \omega^*(p) C_{\omega^*}((p,\overline{p})(q)) \\
&=
  \sum_{x \in D^2} C_{\omega^*}(x)
  +
  \omega^*(p) C_{\omega^*}((p,\overline{p})(q)) \\
&>
  M(\omega^*).
\end{align*}
So $(\omega^*)^2 \in \Omega$ which contradicts that $\omega^*$ is optimal.

\subsection {Proof of \cref{r:get:valuations}}
\label{s:r:get:valuations}
Assume that there is no $\omega \in \fpol^{(2)}(\Gamma,\Delta)$ that is $(a,b)$ or $(b,a)$-dominating.
By \cref{r:use:dom} we know that there are $\nu_1,\nu_2 \in \eclose{\Gamma,\Delta}$ \st 
$\argmin_{x \in D} \nu_1(x) = \{a\}$, $\argmin_{x \in D} \nu_2(x) = \{b\}$
and $\nu_1(x),\nu_2(x)<\infty$ for $x \in \{a,b\}$.
This means, since $\{a,b\} \not\in \Gamma$,
that $\nu_1(x),\nu_2(x) < \infty$ for $x \in D$.
It is not hard to see that, since $\{a,b\} \not\in \Gamma \supseteq \wclose{\Gamma,\Delta}^{(1)}$ we must have 
$\nu_1(a)<\nu_1(c)<\nu_1(a)$ and $\nu_2(b)<\nu_2(c)<\nu_2(a)$
as otherwise there is $\alpha > 0$ \st $\argmin_{x \in D} (\alpha \nu_1(x) + \nu_2(x)) = \{a,b\}$.

\subsection {Proof of \cref{r:two:subsets}}
\label{s:r:two:subsets}
Since $\omega$ is $(a,b,c)$-dominating we have, using \cref{r:measure}, for any instance $I$ of \minhom$(\Gamma,\Delta)$ and any variable $v$ \st $\{a,b\} \subseteq \{\phi(v) : \phi \in \sol(I) \}$,
\begin{align*}
\frac{3 W^\omega_a - 1}{1-3 W^\omega_c} m(\specialsol{I}{v}{a}) + \frac{3 W^\omega_b - 1}{1-3 W^\omega_c} m(\specialsol{I}{v}{b}) &\le m(\specialsol{I}{v}{c}).
\end{align*}
Note that the coefficients in 
the left-hand side
are non-negative and sum to one.

We will show that if there is $\psi\in\fpol^{(2)}(\Gamma,\Delta)$ that is $(a,b)$-dominating or $(b,a)$-dominating,
then either $\{a,c\},\{b,c\} \in \wclose{\Gamma,\Delta}$, or G1 or G3 is true.
Assume we have such a fractional polymorphism $\psi$ and \wlg that $\psi$ is $(a,b)$-dominating.
We have, again with non-negative coefficients summing to one;
\begin{align*}
\frac{2W^\psi_a(a,b)}{1-2W^\psi_b(a,b)} m(\specialsol{I}{v}{a}) + \frac{2W^\psi_c(a,b)}{1-2W^\psi_b(a,b)} m(\specialsol{I}{v}{c}) &\le m(\specialsol{I}{v}{b}).
\end{align*}
Clearly this implies that $m(\specialsol{I}{v}{a}) \le m(\specialsol{I}{v}{c})$ and $m(\specialsol{I}{v}{a}) \le m(\specialsol{I}{v}{b})$.
So $\{a,b,c\}$ is reducible to $\{a,b\}$ and $\{a,b\}$ is reducible to $\{a\}$, so $\{a,b,c\}$ is reducible to $\{a\}$.

If $\{a,c\},\{b,c\} \in \wclose{\Gamma,\Delta}$, then we are done.
Otherwise at most one of $\{a,c\},\{b,c\}$ is in $\wclose{\Gamma,\Delta}$.
It follows from \cref{r:np:hard} that either G1 or G3 is true.

Otherwise there is no $\psi\in\fpol^{(2)}(\Gamma,\Delta)$ that is $(a,b)$-dominating or $(b,a)$-dominating.
By \cref{r:get:valuations} there are $\nu_1,\nu_2 \in \eclose{\Gamma,\Delta}$ \st $\nu_1(a)<\nu_1(c)<\nu_1(a)$ and $\nu_2(b)<\nu_2(c)<\nu_2(a)$.
Consider the following cases.
\begin{itemize}
\item
$\{a,c\},\{c,b\} \in \wclose{\Gamma,\Delta}$

Here we are done.

\item
$\{a,c\} \in \wclose{\Gamma,\Delta}$ and $\{c,b\} \not\in \wclose{\Gamma,\Delta}$

If follows from the existence of $\nu_1,\nu_2$, \cref{r:min:set}, \cref{r:cross} and the fact that
$\{a,b\},\{b,c\} \not\in \wclose{\Gamma,\Delta}$ that either
(1)
$R_1=\{(a,c),(c,a)\} \in \wclose{\Gamma,\Delta}$,
(2)
$R_2=\{(a,b),(b,a),(c,c)\} \in \wclose{\Gamma,\Delta}$, or
(3)
there are set functions $g_1,g_2$ that preserve $\Gamma$ and satisfy $\nu_i(g_i(X)) = \min\{ \nu_i(x):x \in \bigcap_{Y \in \close{\Gamma}:Y \supseteq X} Y\}$ for $i \in [2]$.
\begin{enumerate}
\item
By \cref{r:fpol} we know, since $\{a,b\},\{b,c\} \not\in \wclose{\Gamma,\Delta}$
and because of $\nu_2$,
that there is $\omega \in \fpol^{(2)}(\Gamma,\Delta)$ with $f,g \in \supp(\omega)$ \st $f(b,c)=f(c,b)=b$
and $\{g(a,b),g(b,a)\} \subseteq \{b,c\}$.
From this it follows that $\omega^2$ have operations $f',g' \in \supp(\omega^2)$ \st $f'|_{\{b,c\}}$ and $g'|_{\{a,b\}}$
are commutative and $g'(a,b)\in \{b,c\}$.
This means, by \cref{r:min:commutative}, that there is $\varpi \in \fpol^{(2)}(\Gamma,\Delta)$ \st every $f \in \supp(\varpi)$
and every $x \in D^2$ it holds that either $f|_{x}$ is commutative, or $\{f(x),\overline{f}(x)\}=\{a,c\}$.

\begin{myclaim}
There is $f_1,f_2,f_3 \in \pol^{(2)}(\Gamma)$ \st $f_1 \in \comm{b}{c}{b}$, $f_2 \in\comm{b}{c}{a}$ and $f_3 \in\comm{b}{c}{c}$.
\end{myclaim}
\begin{proof}
There is $g \in \supp(\varpi)$ \st $g \in \comm{b}{c}{b}$.

Since $\{a,b\},\{b,c\} \not\in \close{\Gamma}$ there are $p,q \in \pol^{(2)}(\Gamma)$ \st $p(b,a)=c$ and $q(b,c)=a$.
Because of $\nu_1$ there must be $g \in \supp(\varpi)$ \st either $g(b,c)=g(c,b) \in \{c,a\}$
or $\{g(b,c),g(c,b)\}=\{a,c\}$.
In the first case we can pick $f_2,f_3$ as $q[f_1,g],g$ or $g,p[f_1,g]$.
Consider now the latter case.
Assume \wlg that $g(b,c)=a$ and $g(c,b)=c$.
Here $p[g,\overline{g}] \in\comm{b}{c}{c}$  and $q[f_1,p'] \in\comm{b}{c}{a}$.
\end{proof}

\begin{myclaim}
There is $f_4,f_5 \in \pol^{(2)}(\Gamma)$ \st $f_4 \in\comm{b}{a}{c}$ and $f_5 \in\comm{b}{a}{a}$.
\end{myclaim}
\begin{proof}
Because of $\nu_2$ there must be $g \in \supp(\varpi)$ \st $g(a,b)=g(b,a) \in \{b,c\}$.
Since $\{a,b\} \not\in \close{\Gamma}$ there is $p \in \pol^{(2)}(\Gamma)$ \st $p(a,b)=c$.
So, if $g(a,b)=g(b,a)=b$, then $p'=p[\pr_1,g]$ satisfies $p'(a,b)=c$ and $p'(b,a)=b$.
Now $f_3[p',\overline{p'}] \in\comm{b}{a}{c}$.

Assume \wlg that $f_2|_{\{a,c\}}=\pr_1$.
We can pick $f_5 = f_2[\pr_1,f_4]$.
\end{proof}

\begin{myclaim}
There is $f_6 \in \pol^{(2)}(\Gamma)$ \st $f_6 \in\comm{b}{a}{b}$.
\end{myclaim}
\begin{proof}
Let $R = \bigcap \{ S \in \close{\Gamma}^{(2)} : (a,b),(b,a) \in S \}$.
If $(c,b) \not\in R$, then using $P=R_1 \circ R$ (note that because of $f_4$ we have $(c,c)\in R$) we can choose $\beta$ \st $\pr_{3}( \argmin_{x,y,z \in D : (x,y,z) \in R} (\nu_1(x) + \beta \nu_2(z)) = \{b,c\}$.
This contradict that $\{b,c\}\not\in\wclose{\Gamma,\Delta}$, so $(c,b)\in R$.
This means, because of $f_1$, that $(b,b) \in R$.
Since $R$ is generated from the two tuples $(a,b),(b,a)$ there is some $f_6 \in \pol^{(2)}(\Gamma)$ \st $f_6 \in\comm{b}{a}{b}$.
\end{proof}

\begin{myclaim}
We can assume \wlg that $\{f_i(x,y),f_i(y,x)\} \in \{\{a\},\{b\},\{c\},\{a,c\}\}$ for every $i \in [6]$ and $\{x,y\} \in \{ \{a,b\},\{b,c\} \}$.
\end{myclaim}
\begin{proof}
This follows from the fact that every $f \in \pol^{(2)}(\Gamma)$ is a projection on $\{a,c\}$ and that there are $g,h \in \pol^{(2)}(\Gamma)$ \st $g|_{\{a,b\}}$ and $h|_{\{b,c\}}$ are commutative.
\end{proof}

Let $f=f_3$, $g=f_6$, $h=f_4$, $d=f_1$, $q=f_2$, $r=f_5$ and assume \wlg that all of these operations equal $\pr_1$ on $\{a,c\}$.
Let $m \in \pol^{(3)}(\Gamma)$ be arithmetical on $\{a,c\}$.
By \cref{r:np:hard} such an operation must exist unless G3 is true.

We can construct a pair of term operations that is a tournament pair on $\{a,b\},\{b,c\}$ as follows
(we give different constructions depending on the values of $(h,\overline{h})(b,c)$ and $(g,\overline{g})(b,c)$).
\begin{align*}
(h,\overline{h})(b,c)=(a,c) &
\left\{\begin{aligned}
&h(h(y,x),x)\\
&d(d(y,h(x,h(y,x))),d(x,h(x,h(y,x))))
\end{aligned}\right.\\
(h,\overline{h})(b,c)=(c,a) &
\left\{\begin{aligned}
&m(h(h(y,x),y),h(x,y),h(y,h(y,x)))\\
&d(d(y,m(h(y,x),h(y,h(x,y)),h(x,y))),d(x,m(h(y,x),h(y,h(x,y)),h(x,y))))
\end{aligned}\right.\\
(h,\overline{h})(b,c)=(a,a) &
\left\{
\begin{aligned}
(g,\overline{g})(b,c)=(a,c) &
\left\{\begin{aligned}
&h(h(y,x),x)\\
&d(g(g(x,y),y),g(g(y,x),x))
\end{aligned}\right.\\
(g,\overline{g})(b,c)=(c,a) &
\left\{\begin{aligned}
&h(h(y,x),x)\\
&d(g(x,g(x,y)),g(y,g(y,x)))
\end{aligned}\right.\\
(g,\overline{g})(b,c)=(a,a) &
\left\{\begin{aligned}
&h(h(y,x),x)\\
&g(g(y,g(x,y)),g(x,g(y,x)))
\end{aligned}\right.\\
(g,\overline{g})(b,c)=(b,b) &
\left\{\begin{aligned}
&h(h(y,x),x)\\
&g(y,x)
\end{aligned}\right.\\
(g,\overline{g})(b,c)=(c,c) &
\left\{\begin{aligned}
&h(h(y,x),x)\\
&d(d(x,g(y,g(x,h(x,y)))),g(y,x))
\end{aligned}\right.
\end{aligned}\right.\\
(h,\overline{h})(b,c)=(b,b) &
\left\{\begin{aligned}
&h(q(y,h(x,y)),q(h(y,x),x))\\
&h(h(y,h(x,y)),h(x,h(y,x)))
\end{aligned}\right.\\
(h,\overline{h})(b,c)=(c,c) &
\left\{\begin{aligned}
&m(f(y,h(y,x)),h(x,y),f(x,h(y,x)))\\
&d(d(y,h(x,y)),d(x,h(y,x)))
\end{aligned}\right.
\end{align*}
This establishes G1.

\item
Since $\{a,c\} \in \wclose{\Gamma,\Delta}$ this would imply $\{b,c\}\in\wclose{\Gamma,\Delta}$.
So this case is not possible.
\item
Note that with $f_i(x,y)=g_i(\{x,y\})$ we have commutative operations $f_1,f_2 \in \pol^{(2)}(\Gamma,\Delta)$
\st 
$f_1 \in \comm{b&c&b}{a&a&c}{a&a&a}$ and $f_2 \in \comm{b&c&b}{a&a&c}{b&c&b}$.
Since $\{a,b\} \not\in \close{\Gamma}$ there is $p \in \pol^{(2)}(\Gamma)$ \st $p(a,b)=c$.
This means that $f_3=p[f_1,f_2]\in \pol^{(2)}(\Gamma) \cap \comm{b&c&b}{a&a&c}{c&x&c}$ for some $x \in D$.
\begin{itemize}
\item
If $x=c$, set $f'(x,y)=f_1(f_3(f_1(x,y),y),f_3(f_1(x,y),x))$.
\item
If $x=b$, set $f'(x,y)=f_2(f_1(f_3(x,y),y),f_1(f_3(x,y),x))$.
\item
If $x=a$, set $f'(x,y)=f_1(f_3(f_3(x,y),y),f_3(f_3(x,y),x))$.
\end{itemize}
In all cases $f' \comm{b&c&b}{a&a&c}{a&a&c}$.
So $f',f_2$ is a tournament pair,
and G1 is true.
\end{enumerate}

\item
$\{a,c\} \not\in \wclose{\Gamma,\Delta}$ and $\{c,b\} \in \wclose{\Gamma,\Delta}$

Symmetric to the previous cases.

\item
$\{a,c\},\{c,b\} \not\in \wclose{\Gamma,\Delta}$

By \cref{r:fpol} we know, since no two-element subset of $D$ is in $\wclose{\Gamma,\Delta}$
and because of $\nu_1,\nu_2$,
that there is $\omega \in \fpol^{(2)}(\Gamma,\Delta)$ with $f,g,h \in \supp(\omega)$ \st $f(a,c)=f(c,a)=a$, $g(b,c)=g(c,b)=b$
and $\{h(a,b),h(b,a)\} \ne \{a,b\}$.
From this it follows that $\omega^2$ have operations $f',g',h' \in \supp(\omega^2)$ \st $f'|_{\{a,c\}}$, $g'|_{\{b,c\}}$
and $h'|_{\{a,b\}}$ are commutative.
This means, by \cref{r:min:commutative}, that there is $\varpi \in \fpol^{(2)}(\Gamma,\Delta)$ \st every $f \in \supp(\varpi)$ is commutative.

Note that, because of $\nu_1$, we must have $W^\varpi_a(a,c) > 0$
and, because of $\nu_2$, we have $W^\varpi_b(b,c) > 0$.

If follows from the existence of $\nu_1,\nu_2$, \cref{r:min:set}, \cref{r:cross} and the fact that
$\{a,b\},\{a,c\},\{b,c\} \not\in \wclose{\Gamma,\Delta}$ that either
(1)
$R=\{(a,b),(b,a),(c,c)\} \in \wclose{\Gamma,\Delta}$ or
(2)
there are set functions $g_1,g_2$ that preserve $\Gamma$ and satisfy $\nu_i(g_i(X)) = \min\{ \nu_i(x):x \in \bigcap_{Y \in \close{\Gamma}:Y \supseteq X} Y\}$ for $i \in [2]$.
\begin{enumerate}
\item
Note that every operation $h \in \supp(\varpi)$ is commutative and therefore must satisfy $h(a,b)=c$ (otherwise $h$ does not preserve $R$).
We know that there is some $f \in \supp(\varpi)$ \st $f(a,c)=a$.
Since $f$ must preserve $R$ (as $\wclose{\Gamma,\Delta}^{(2)} \subseteq \Gamma$)
it holds that $f(b,c)=b$.

There must also be some $g \in \supp(\varpi)$ \st $g(a,c) \ne a$.
We have two cases to consider.
\begin{itemize}
\item $g(a,c)=c$

Since $g$ preserves $R$ it holds that $g(b,c)=c$.
Now $g$ is a semilattice operation as $g(a,b)=c$, so G2 is true.

\item $g(a,c)=b$

Since $g$ preserves $R$ it holds that $g(b,c)=a$.
Clearly 
\begin{align*}
\sum_{h \in \pol^{(2)}(\Gamma)} \varpi(h) \nu( h[f,g](x,y) ) \le \frac{1}{2} ( \nu(f(x,y)) + \nu(g(x,y)) )
\end{align*}
holds for every $x,y \in D$.
This means that there is another binary fractional polymorphism $\varpi'$ \st $f[f,g] \in \supp(\varpi')$ and \st every $h \in \supp(\varpi')$ is commutative.
Since $f[f,g]$ is a semilattice operation if follows that G2 holds.
\end{itemize}

\item
Note that with $f_i(x,y)=g_i(\{x,y\})$ we have $f_1,f_2 \in \pol^{(2)}(\Gamma,\Delta)$ and $f_1(x,y)=a$ if $x \ne y$ and $f_2(x,y)=b$ if $x \ne y$.
Since $\{a,b\} \not\in \close{\Gamma}$ there is $p \in \pol^{(2)}(\Gamma)$ \st $p(a,b)=c$.
This means that $f_3=p[f_1,f_2]\in \pol^{(2)}(\Gamma)$ satisfies $f_3(x,y)=c$ if $x\ne y$.
Define $p,q$ through $p(x,y)=f_3(f_1(f_3(x,y),y),f_1(f_3(x,y),x))$ and $q(x,y)=f_2(f_1(g(x,y),y),f_1(f_2(x,y),x))$.
It can be checked that $p,q$ is a tournament pair, so G1 is true.
\end{enumerate}
\end{itemize}

\subsection {Proof of \cref{r:twoset:dom}}
\label{s:r:twoset:dom}
Note that $\{a,b\}$ is shrinkable to $\{a\}$ and $\{a,b,c\}$ is shrinkable to $\{a,c\}$.
Consider the following cases.
\begin{itemize}
\item
There is $\psi \in \fpol^{(2)}(\Gamma,\Delta)$ that is $(a,c)$-dominating or $(c,a)$-dominating.
\begin{itemize}
\item
There is $\xi \in \fpol^{(2)}(\Gamma,\Delta)$ that is $(b,c)$-dominating or $(c,b)$-dominating.

Here $\{a,c\}$ is shrinkable to either $\{a\}$ or $\{c\}$ and 
$\{b,c\}$ is shrinkable to either $\{b\}$ or $\{c\}$, so G1 holds.

\item
There is no $\xi \in \fpol^{(2)}(\Gamma,\Delta)$ that is $(b,c)$-dominating or $(c,b)$-dominating.

From \cref{r:use:dom} it follows that there are $\nu_1,\nu_2 \in \eclose{\Gamma,\Delta}$ \st 
$\nu_1(x),\nu_2(x)< \infty$ for $x \in \{b,c\}$,
$\argmin_{x \in D} \nu_1(x) = \{b\}$ and $\argmin_{x \in D} \nu_2(x) = \{c\}$.
Now \cref{r:np:hard} implies that either G3 or G1 holds.

\end{itemize}

\item
There is no $\psi \in \fpol^{(2)}(\Gamma,\Delta)$ that is $(a,c)$-dominating or $(c,a)$-dominating.
\begin{itemize}
\item
There is $\xi \in \fpol^{(2)}(\Gamma,\Delta)$ that is $(b,c)$-dominating or $(c,b)$-dominating.

This case is symmetric to the last.
\item
There is no $\xi \in \fpol^{(2)}(\Gamma,\Delta)$ that is $(b,c)$-dominating or $(c,b)$-dominating.

From \cref{r:use:dom} it follows that there are $\nu_1,\nu_2,\nu_3,\nu_4 \in \eclose{\Gamma,\Delta}$ \st 
$\nu_1(x),\nu_2(x)< \infty$ for $x \in \{a,c\}$, $\nu_3(x),\nu_4(x)< \infty$ for $x \in \{b,c\}$,
$\argmin_{x \in D} \nu_1(x) = \{a\}$, $\argmin_{x \in D} \nu_2(x) = \{c\}$,
$\argmin_{x \in D} \nu_3(x) = \{b\}$ and $\argmin_{x \in D} \nu_4(x) = \{c\}$.

\begin{itemize}
\item
$\{ (a,c),(c,a) \} \in \wclose{\Gamma,\Delta}$ and $\{ (b,c),(c,b) \} \in \wclose{\Gamma,\Delta}$

From \cref{r:np:hard} it follows that, unless G3, there must be $m_1,m_2 \in \pol^{(3)}(\Gamma)$ \st $m_1|_{\{a,c\}}$ is arithmetical and
$m_2|_{\{b,c\}}$ is arithmetical. 
By \cref{r:arithmetical} we know that G1 is true.

\item
$\{ (a,c),(c,a) \} \in \wclose{\Gamma,\Delta}$ and $\{ (b,c),(c,b) \} \not\in \wclose{\Gamma,\Delta}$

From \cref{r:fpol} we know that there is some $\kappa \in \fpol^{(2)}(\Gamma,\Delta)$ 
with $f,g \in \supp(\kappa)$ \st $f\in\comm{c}{b}{c}$ and $g\in\comm{c}{b}{b}$.
From \cref{r:np:hard} it follows that, unless G3, there must be $m \in \pol^{(3)}(\Gamma)$ \st $m_1|_{\{a,c\}}$ is arithmetical. 
This implies G1.

\item
$\{ (a,c),(c,a) \} \not\in \wclose{\Gamma,\Delta}$ and $\{ (b,c),(c,b) \} \in \wclose{\Gamma,\Delta}$

Symmetric to the previous.
\item
$\{ (a,c),(c,a) \} \not\in \wclose{\Gamma,\Delta}$ and $\{ (b,c),(c,b) \} \not\in \wclose{\Gamma,\Delta}$

From \cref{r:fpol} we know that there some $\kappa \in \fpol^{(2)}(\Gamma,\Delta)$ 
with $f,g \in \supp(\kappa)$ \st $f\in\comm{c}{b}{b}$ and $g\in\comm{c}{a}{a}$.
By \cref{r:min:commutative} and the fact that $\{a,c\},\{b,c\} \in \wclose{\Gamma,\Delta}$
there is $\kappa' \in \fpol^{(2)}(\Gamma,\Delta)$ \st every $f \in \supp(\kappa')$ is commutative on $\{a,c\}$ and $\{b,c\}$.
We know that $W^{\kappa'}_a(a,c)=W^{\kappa'}_c(a,c)=\frac{1}{2}$ and 
$W^{\kappa'}_b(b,c)=W^{\kappa'}_c(b,c)=\frac{1}{2}$.
So there must be a pair $f,g \in \supp(\kappa')$ that is a tournament pair on $\{\{a,c\},\{c,b\}\}$.
This means that G1 is true.
\end{itemize}

\end{itemize}
\end{itemize}

\subsection {Proof of \cref{r:work:a}}
\label{s:r:work:a}
By \cref{r:twoset:dom} we can assume that there is no $\omega \in \fpol^{(2)}(\Gamma,\Delta)$ that is $(b,a)$-dominating or $(b,a)$-dominating.
By \cref{r:get:valuations} there are $\nu_1,\nu_2 \in \eclose{\Gamma,\Delta}$ \st $\nu_1(a)<\nu_1(c)<\nu_1(a)$ and $\nu_2(b)<\nu_2(c)<\nu_2(a)$.

If follows from \cref{r:min:set}, \cref{r:cross} and the fact that
$\{a,b\} \not\in \wclose{\Gamma,\Delta}$ and $\{a,c\},\{c,b\} \in \wclose{\Gamma,\Delta}$ that either
(1)
$R_1=\{(a,c),(c,a)\} \in \wclose{\Gamma,\Delta}$ or $R_2=\{(b,c),(c,b)\} \in \wclose{\Gamma,\Delta}$,
(2)
$R_2=\{(a,b),(b,a),(c,c)\} \in \wclose{\Gamma,\Delta}$, or
(3)
there are set functions $g_1,g_2$ that preserve $\Gamma$ and satisfy $\nu_i(g_i(X)) = \min\{ \nu_i(x):x \in \bigcap_{Y \in \close{\Gamma}:Y \supseteq X} Y\}$ for $i \in [2]$.
\begin{enumerate}
\item
In this case we are done.
\item
By \cref{r:fpol} we know, since we can assume (i) is false, since $\{a,b\} \not\in \wclose{\Gamma,\Delta}$ and because of $\nu_1,\nu_2$,
that there is $\omega \in \fpol^{(2)}(\Gamma,\Delta)$ with $f,g,h \in \supp(\omega)$ \st
$f\in\comm{c}{a}{a}$, $g\in\comm{c}{b}{b}$ and $\{h(a,b),h(b,a)\} \ne \{a,b\}$.
From this it follows that $\omega^2$ have operations $f',g',h' \in \supp(\omega^2)$ \st $f'|_{\{a,c\}}$, $g'|_{\{b,c\}}$
and $h'|_{\{a,b\}}$ are commutative.
This means, by \cref{r:min:commutative}, that there is $\varpi \in \fpol^{(2)}(\Gamma,\Delta)$ \st every $f \in \supp(\varpi)$ is commutative.

Since every operation $f \in\supp(\varpi)$ must preserve $R_2$ we know that 
$f \in \comm{b}{a}{c}$ and that if $f \in\comm{c}{a}{c}$, then $f \in\comm{c}{b}{c}$.
Note that, by $\nu_2$ there must be some $g \in \supp(\varpi)$ \st $g \in\comm{c}{a}{c}$.
It follows that $g$ is a semilattice operation, so G2 holds.

\item
Recall that $\wclose{\Gamma,\Delta}^{(1)} \subseteq \Gamma$.
With $f_i(x,y)=g_i(\{x,y\})$ we have $f_1,f_2 \in \pol^{(2)}(\Gamma,\Delta)$ and 
$f_1,f_2$ equals the $\min,\max$ with respect to the ordering $a < c < b$.
Clearly $f_1,f_2$ is a tournament pair, so G1 holds.
\end{enumerate}

\subsection {Proof of \cref{r:class:one}}
\label{s:r:class:one}
To prove the lemma we will make use of the following observations.
\begin{mylemma}
\label{r:constructions}
Let $D$ be any set.
Assume $f,g,h \in \opers 2$ and $m \in \opers 3$ are idempotent and that $a,b,c \in D$ are distinct.
\begin{enumerate}
\item
\label{construction:make:arithmetical}
If $m$ is arithmetical on $\{ \{b,c\}, \{c,a\} \}$, $f|_{\{b,c\}}=\pr_1$, $f|_{\{c,a\}}=\pr_1$ and $f(b,a)=f(a,b)=c$, then
$m'(x,y,z)=m(m(f(x,z),f(y,x),x),$ $f(x,z),$ $m(z,f(y,z),f(x,z)))$ is arithmetical on $\{ \{b,c\}, \{b,a\}, \{c,a\} \}$.
\item
\label{construction:make:ac:comm:b}
If $f|_{\{b,c\}}=g|_{\{b,c\}}=\pr_1$, $f|_{\{c,a\}}=g|_{\{c,a\}}=\pr_1$, $\{f(b,a),f(a,b)\}=\{b,c\}$ and $\{g(b,a),g(a,b)\}=\{c,a\}$,
then either $f'=f[f,\pr_1]$, $f'=f[g,\pr_1]$ or $f'=f[g,\pr_2]$ satisfies $f'(b,a)=f'(a,b)=c$.
\item
\label{construction:use:arith:b}
If $f \in\comm{c&a}{b&b}{b&c}$, $g \in\comm{c&a}{b&b}{c&c}$ and $f|_{\{c,a\}}=g|_{\{c,a\}}=\pr_1$.
Then $(x,y) \mapsto f(f(f(x,y),y),f(f(x,y),x)) \in\comm{c&a}{b&b}{b&b}$
and $(x,y) \mapsto m(g(y,g(y,x)),g(x,y),g(x,g(y,x))) \in\comm{c&a}{b&b}{c&a}$.
\end{enumerate}
\end{mylemma}

By \cref{r:np:hard} there is, unless G3, an operation $m \in \pol^{(3)}(\Gamma)$ that is arithmetical on $\{\{a,c\}\}$.
By \cref{r:twoset:dom} we can assume that there is no $\omega \in \fpol^{(2)}(\Gamma,\Delta)$ that is $(b,a)$-dominating or $(b,a)$-dominating.
By \cref{r:get:valuations} there are $\nu_1,\nu_2 \in \eclose{\Gamma,\Delta}$ \st $\nu_1(a)<\nu_1(c)<\nu_1(a)$ and $\nu_2(b)<\nu_2(c)<\nu_2(a)$.

By \cref{r:fpol} and the existence of $\nu_2$ there is $\omega \in \fpol^{(2)}(\Gamma,\Delta)$ with $f,h \in \supp(\omega)$ \st
$f(b,c)=f(c,b)$ and $\{h(a,b),h(b,a)\} \subseteq \{b,c\}$.
This means that there are $f',h' \in \supp(\omega^2)$ \st $f'(b,c)=f'(c,b)$ and $h'(a,b)=h'(b,a)$.
By \cref{r:min:commutative} there is $\psi \in \fpol^{(2)}(\Gamma,\Delta)$ \st $(s,\overline{s})(x) \not\in 
\{ (a,b),(b,a),(b,c),(c,b) \}$ for every $s \in \supp(\psi)$ and $x \in D^2$.

Because of $\nu_1,\nu_2$ it holds that $W_b^\psi(b,c)=W_c^\psi(b,c)=\frac{1}{2}$.
\begin{itemize}
\item
If $W_c^\psi(a,b)=0$, then since $\psi$ is not $(a,b)$ or $(b,a)$-dominating we must have $W_a^\psi(a,b)=W_b^\psi(a,b)=\frac{1}{2}$.
In this case there must be $f,g \in \supp(\psi)$ that is a tournament pair on $\{\{a,b\},\{b,c\}\}$,
and G1 is true.
\item
If $W_a^\psi(a,b)=W_b^\psi(a,b)=0$, then there is $f,g \in \supp(\psi)$ \st
$f \in\comm{c&a}{b&b}{b&c}$ and $g \in\comm{c&a}{b&b}{c&c}$.
By \cref{r:constructions}\eqref{construction:use:arith:b} there is again a tournament pair on $\{\{a,b\},\{b,c\}\}$,
and G1 is true.
\item
If $W_a^\psi(a,b)=0, W_b^\psi(a,b)>0$,
then, since $\psi$ is not $(b,a)$-dominating, it holds that $W_c^\psi(a,b)>\frac{1}{2}$.
So again there are $f,g \in \supp(\psi)$ \st
$f \in\comm{c&a}{b&b}{b&c}$ and $g \in\comm{c&a}{b&b}{c&c}$.
As in the previous case G1 holds.
\item
If $W_a^\psi(a,b)>0, W_b^\psi(a,b)=0$, 
note that $W_a^\psi(a,b) < W_c^\psi(a,b)$.
Otherwise, by $\nu_2$, the languages $(\Gamma,\Delta)$ can not admit $\psi$.
Hence, there must be $f \in \supp(\psi)$ \st $f\in\comm{a&b}{b&c}{c&x}$ where $x \in \{b,c\}$.
Define $\overline{x}$ \st $\{x,\overline{x}\}=\{b,c\}$.
Every $g \in \supp(\psi) \cap \comm{b}{c}{\overline{x}}$ can not satisfy $g\in\comm{a}{b}{a}$ as that would
imply that $\psi$ is $(a,b)$-dominating.
If some $g \in \supp(\psi) \cap \comm{b}{c}{\overline{x}}$ satisfies $g\in\comm{a}{b}{c}$,
then using \cref{r:constructions}\eqref{construction:use:arith:b} with $f,g$ we see there is a tournament pair on $\{\{a,b\},\{b,c\}\}$,
and G1 is true.
Hence, there must be some $g \in \supp(\psi) \cap \comm{b}{c}{\overline{x}}$
\st $\{g(a,b),g(b,a)\}=\{a,c\}$.
Assume \wlg that $g|_{\{a,c\}}=\pr_2$.
Note that $g[g,f] \in\comm{a&b}{b&c}{c&\overline{x}}$.
By \cref{r:constructions}\eqref{construction:use:arith:b} there is again a tournament pair on $\{\{a,b\},\{b,c\}\}$,
and G1 is true.

\item
Otherwise, $W_a^\psi(a,b)>0, W_b^\psi(a,b)>0, W_c^\psi(a,b)>0$.

We know that there is $f \in \supp(\psi)$ \st $f\in\comm{a&b}{b&c}{b&x}$ for some $x \in \{b,c\}$.

Note that, since $\psi$ is not $(b,a)$-dominating, there is some operation
$h \in \supp(\psi)$ \st $h\in\comm{b}{c}{x}$ \st $h(a,b)=h(b,a) \in \{a,c\}$ or $\{h(a,b),h(b,a)\}=\{a,c\}$.
If $h(a,b)=h(b,a)=a$, then since $\{a,b\}\not\in\close{\Gamma}$ there is $p \in\pol^{(2)}(\Gamma)$ \st $p(a,b)=c$.
This means that $p[h,f]\in\comm{a&b}{b&c}{c&x}$.
If $\{h(a,b),h(b,a)\}=\{a,c\}$, then assume \wlg that $h(a,b)=c$.
This means that $h'=h[h,f] \in\comm{b}{c}{x}$ satisfies $h'(a,b) \in \{c,b\}$ and $h'(b,a)=c$,
and that, since $W_c^\psi(b,c)=\frac{1}{2}$, there is $h'' \in \pol^{(2)}(\Gamma)$ \st $h''\in\comm{a&b}{b&c}{c&x}$.
So, we can assume \wlg that $h\in\comm{a&b}{b&c}{c&x}$.

Define $y$ \st $\{x,y\}=\{b,c\}$.
Since $\psi$ is not $(b,a)$-dominating there is some operation
$g \in \supp(\psi)$ that is not in $\comm{a&b}{b&c}{b&y}$.
If $g\in\comm{a&b}{b&c}{a&y}$,
then G1 holds.
If $g\in\comm{a&b}{b&c}{c&y}$, then by 
\cref{r:constructions}\eqref{construction:use:arith:b} there is a tournament pair on $\{\{a,b\},\{b,c\}\}$,
and G1 is true.
Otherwise $\{g(a,b),g(b,a)\}=\{a,c\}$.
Assume \wlg that $g|_{\{a,c\}}=\pr_2$.
Note that $g[g,h] \in\comm{a&b}{b&c}{c&y}$.
By \cref{r:constructions}\eqref{construction:use:arith:b} there is again a tournament pair on $\{\{a,b\},\{b,c\}\}$,
and G1 is true.
\end{itemize}

\subsection {Proof of \cref{r:work:b}}
\label{s:r:work:b}
By \cref{r:np:hard} and \cref{r:arithmetical} there is, unless \minhom$(\Gamma,\Delta)$ is \class{NP}-hard, an operation $m \in \pol^{(3)}(\Gamma)$ that is arithmetical on $\{\{a,c\},\{c,b\}\}$.
By \cref{r:twoset:dom} we can assume that there is no $\omega \in \fpol^{(2)}(\Gamma,\Delta)$ that is $(b,a)$-dominating or $(b,a)$-dominating.
By \cref{r:get:valuations} there are $\nu_1,\nu_2 \in \eclose{\Gamma,\Delta}$ \st $\nu_1(a)<\nu_1(c)<\nu_1(a)$ and $\nu_2(b)<\nu_2(c)<\nu_2(a)$.
By \cref{r:fpol} and the existence of $\nu_1,\nu_2$ there is $\omega \in \fpol^{(2)}(\Gamma,\Delta)$ with $f,g \in \supp(\omega)$ \st
$\{f(a,b),f(b,a)\} \subseteq \{a,c\}$ and $\{g(a,b),g(b,a)\} \subseteq \{c,b\}$.
\begin{itemize}
\item
If $\{f(a,b),f(b,a)\}=\{c\}$ or $\{g(a,b),g(b,a)\}=\{c\}$, then by (using $f[f,\overline{f}]$ or $g[g,\overline{g}]$)
\cref{r:constructions}\eqref{construction:make:arithmetical} implies G1.
\item
If $\{f(a,b),f(b,a)\}=\{a\}$ and $\{g(a,b),g(b,a)\}=\{b\}$, then $f,g$ is a tournament pair on $\{\{a,b\}\}$,
so G1 is true.
\item
If $\{f(a,b),f(b,a)\}=\{a,c\}$ and $\{g(a,b),g(b,a)\}=\{b,c\}$, then by (using $f[f,\overline{f}]$ or $g[g,\overline{g}]$)
\cref{r:constructions}\eqref{construction:make:ac:comm:b} 
we know that there is $h \in \pol^{(2)}(\Gamma)$ \st $h(a,b)=h(b,a)=c$.
This brings us to the first case.
\item
If $\{f(a,b),f(b,a)\}=\{a,c\}$ and $\{g(a,b),g(b,a)\}=\{b\}$,
then assume \wlg that $f|_{\{b,c\}}=\pr_1$.
If $f(a,b)=a$, then $f'=f[g,f]$ satisfies $f'(a,b)=c$ and $f'(b,a)=b$.
This takes us to the previous case.
Otherwise $f(a,b)=c$.
Here $f'=f[f,g]$ satisfies $f'(a,b)=c$ and $f'(b,a)=c$, which takes us to the first case.
\item
Otherwise $\{f(a,b),f(b,a)\}=\{a\}$ and $\{g(a,b),g(b,a)\}=\{c,b\}$.
This case is handled in a way symmetrical to the previous.
\end{itemize}

\end{document}